\documentclass{elsarticle}
\makeatletter
\def\ps@pprintTitle{%
	\let\@oddhead\@empty
	\let\@evenhead\@empty
	\def\@oddfoot{}%
	\let\@evenfoot\@oddfoot}
\makeatother

\usepackage[utf8]{inputenc}
\usepackage[T1]{fontenc}
\usepackage{rsfso}
\usepackage{amsthm}
\usepackage{newtxtext,newtxmath} 
\usepackage{setspace}
\usepackage{graphicx}
\usepackage{subfig}
\usepackage{caption}
\usepackage{amsmath}
\theoremstyle{plain}
\usepackage{algorithm}
\usepackage{algpseudocode}

\usepackage{tikz}
\usetikzlibrary{arrows.meta, positioning, calc}
\usepackage{booktabs,tabularx,array}
\usepackage[hidelinks]{hyperref}
\newcolumntype{Y}{>{\raggedright\arraybackslash}X}
\newtheorem{theorem}{Theorem}[section]

\newtheorem{remark}{Remark}[section]

\usepackage{color}
\usepackage{colordvi}

\newtheorem{prop}{Proposition}[section]

\newtheorem{cor}[prop]{Corollary}

\newcommand{\komment}[1]{}

\usepackage[right]{lineno}

\begin{document}

\begin{frontmatter}
	\title{Multi-type branching inference on contact trees with application to COVID-19}

	\author[TUMGarching]{Augustine Okolie\corref{mycorrespondingauthor}}
	\author[TUMGarching,ICB]{Johannes M\"uller}
        \author[UNILAG]{Eno Akarawak}
         \author[UNILAG]{Isaac Ajiboye}
	\address[TUMGarching]{Center for Mathematical Sciences, Technische Universit\"at M\"unchen, 85748 Garching, Germany}
	\address[ICB]{Institute for Computational Biology, Helmholtz Centre Munich, 85764 Neuherberg, Germany}
	\address[UNILAG]{University of Lagos, University Road, Lagos Mainland, Lagos, Nigeria}

\cortext[mycorrespondingauthor]{Corresponding author.\\
 Email address: \url{augustine.okolie@tum.de} (A.Okolie).}


\begin{abstract}
Inferring epidemiological parameters from transmission trees is essential for understanding infectious disease dynamics. Existing tree-based likelihood
methods, including the multi-type birth-death models originally applied in
phylodynamic settings, provide powerful tools, but most assume homogeneous mixing and rarely capture how transmission potential changes as an individual
infects more of their contacts. In this work, we develop a likelihood framework
that operates directly on transmission trees, in which nodes are individuals and edges are reported transmission events, with no sequence data
involved. We derive a likelihood for a stochastic SIR
process on a rooted contact tree in which each infected individual is characterised by the total number of effective contacts, and the number of already infected downstream contacts. We obtain closed-form ordinary differential equations for the probability that a clade goes entirely unobserved and for the probability density that it produces an observed (sampled) tip in a given state. The resulting likelihood can be evaluated for a
rooted contact tree with known tip states, and we extend it to partially resolved trees by treating internal branching times as latent variables.
Validation on simulated outbreaks confirms accurate parameter recovery and well
calibrated uncertainty. Application to empirical COVID-19 contact-tracing data
from Karnataka, India, demonstrates the framework's utility for real
epidemiological settings. By incorporating contact-degree heterogeneity in a multi-type branching likelihood, our work provides a principled baseline for
inferring both transmission dynamics and contact structure from fully or
partially resolved transmission trees, complementing rather than relying on
sequence-based phylodynamic inference.
\end{abstract}

\end{frontmatter}

\section{Introduction}
Infectious disease dynamics represent a critical area of study in epidemiology, requiring robust methodologies to understand how diseases spread within populations and how their transmission parameters can be inferred. Several major methods in this domain rely heavily on multi-type branching processes, which provide a natural framework for modelling transmission in heterogeneous populations where individuals may differ in their disease state, contact structure or other relevant characteristics \citep{rasmussen2017phylodynamics, laha2023multi, zhao2024impacts, ball2017heterogeneous, ball2026multitype, singh2014using}. The application of these models to epidemiology was substantially advanced by \citet{stadler2013uncovering}, who developed a maximum likelihood framework based on a multitype birth-death branching (MTBD) process to infer parameters in structured populations. In their framework, each individual belongs to a distinct transmission group, allowing estimation of epidemiological parameters based on phylogenetic trees that depend on type, such as transmission and recovery rates. For these models, the likelihood is computed on a rooted tree by propagating, along the branches of the tree, a pair of backward differential equations. The likelihood central to these tree-based methods originates in the state-dependent birth-death model of \citet{maddison2007estimating}, who evaluated the effects of character types on diversification rates. The approach has been applied to viral phylogenies to uncover the role of structured populations in shaping epidemic trajectories, including studies of HIV and influenza \citep{grenfell2004unifying, kuhnert2014simultaneous, volz2009phylodynamics, volz2013viral, kuhnert2014simultaneous, stadler2013uncovering,stadler2013skyline}.

\medskip
The underlying likelihood framework we use is close to these phylogenetic methods, but with no sequence data involved. Specifically, multi-type branching processes evaluated on a rooted binary tree were originally devised for phylogenies inferred from pathogen genome sequences; however, the same framework extends naturally to contact tracing data, where the rooted binary tree is built directly from reported transmission events. In that setting, nodes correspond to individuals, edges to reported transmission events, and tips to observed cases, and no molecular sequence information enters the analysis. Our contribution is therefore to develop a tree likelihood specifically for contact-traced transmission trees, building on the multi-type branching process tradition used in phylodynamics but applied directly to contact tracing data without sequence input.

\medskip
Recent work has expanded the scope of multi-type branching process models to address key practical challenges in epidemic modeling. \citet{zhukova2025accounting} extended MTBD models to account for contact tracing, where detected individuals notify their contacts who are then rapidly tested and isolated. They developed a simulator and a nonparametric test for detecting contact tracing signals in transmission trees, and showed that ignoring this sampling mechanism leads to severe bias in parameter estimates, particularly overestimation of the removal rate. Separately, \citet{laha2023multi} proposed a time varying multi-type branching process model that captures the effects of government interventions including lockdown, testing, and partial contact tracing. Their model distinguishes between detected and undetected cases, tracks quarantined individuals identified through contact tracing, and estimates the number of unreported spreaders and nonspreaders in the population.

\medskip
Despite these advances, existing multi-type branching process models typically assume homogeneous mixing and that transmission rates depend only on the type or state of the individual, not on the history of transmissions that individual has already made. In reality, an individual's remaining transmission potential decreases as they infect more of their contacts, since each infection event depletes the pool of susceptible contacts. This saturation effect is particularly relevant for diseases where contact structure matters, such as sexually transmitted infections or household-based transmission, and is well documented in network epidemiology \citep{miller2012edge,kenah2007network,istvan2017mathematics}. \citet{metzig2019phylogenies} examined the role of contact structure using phylogenies generated from dynamic networks, although their study relied entirely on simulation-based analyses. A notable step
towards incorporating contact structure is the work of \citet{rasmussen2017phylodynamics}, who introduced a pairwise coalescent model
that incorporates host degree into a likelihood framework. Their approach uses a random graph model and stratifies
the population by the number of contacts each individual has, tracking connections between different contact classes. The model still
assumes, however, that an infected individual transmits at a rate proportional to the total number of contacts, regardless of how many of those contacts have already been infected.

\medskip
Our approach builds upon and consolidates this existing knowledge by developing a multi-type branching process model defined directly on a rooted contact tree \citep{okolie2020exact, okolie2023parameter, okebunor2022contact}, where each individual has a degree drawn from a degree distribution representing their number of contacts. We stratify the population by both contact number and infection stage, but crucially extend the type space to include a dynamic state that tracks how many of an individual's downstream contacts have already been infected. This allows the transmission rate to decrease as the infection progresses and susceptible contacts are depleted, capturing the reduction in transmission potential that arises from contact depletion. We derive the likelihood for this model and apply it in a Bayesian inference framework to contact tracing data from the first wave of COVID-19 in Karnataka, India. Our analysis recovers reproduction numbers consistent with independent contact tracing studies while providing new insight into how contact network shape the structure of the underlying transmission tree.

\section{Model Setup}\label{sec:model}
We consider a stochastic susceptible–infected–recovered (SIR) model on a rooted contact tree in which the root node is the primary infected individual, and every node represents an individual in the population \cite{okolie2020exact, okolie2023parameter}. Each node has a downstream degree $k$, drawn independently and identically as a realisation of a random variable $K$ with probability mass function $w_k:=\mathbb{P}(K=k)$, $k\geq 0$, for which we write $K\sim w_k$, where the downstream degree records the total number of contacts available for onward transmission. Downstream contacts are the edges from an infector to its potential infectees (i.e., the direction of epidemic spread away from the root); equivalently, a node is downstream of another if the path from the root passes through the latter \cite{okolie2020exact}. At each time point every node carries an SIR state in $\{S,I,R\}$, with all nodes initially in state $S$ except the root, which is in state $I$. A node in state $I$ infects each of its still susceptible downstream neighbours independently at rate $\beta$, and leaves the infectious class either observed ($\sigma$) or unobserved ($\mu$), at a total rate $\gamma=\mu + \sigma $. With probability $p_{\mathrm{obs}}=\sigma/\gamma$, this removal is observed and generates a sampled tip carrying the individual's current state, while with probability $1-p_{\mathrm{obs}}=\mu/\gamma$ it is unobserved and the individual leaves no record.

\medskip
We augment the disease state of a node by its character type $(i,k)\in\{0,\dots,k\}\times\{0,1,\dots\}$ denoting the number of
already infected downstream contacts $i$ out of it total downstream contacts $k$. Here $i$ records how many of the $k$ downstream edges have already been used to transmit, so
$k-i$ is the number of downstream edges still leading to a susceptible
contact. At the moment of infection, a node enters type $(0,k)$, and each subsequent transmission along one of its $k-i$ susceptible edges
advances its type from $(i,k)$ to $(i+1,k)$. A newly infected
always has type $(0,K')$ with $K'$ drawn independently from
$w_k$, and upon recovery (observed or unobserved), the node exits
the active class into the recovered state $R$ with final type $(i^{\ast},k)$, where $0\le i^{\ast}\le k$.
Under this augmentation, the effective infection rate of a node with type $(i,k)$ is $(k-i)\beta$. The data analysed by our likelihood is a sampled subtree of the
underlying SIR process, in its most informative form a rooted binary
topology with the calendar time of every internal transmission event and observed tip together with the observed type at each tip. We treat
the degree $k$ as latent throughout, with distribution $w_k$, so the framework applies whether or not tip degrees are recorded; when they
are recorded, they supply additional information for inference.

\medskip
For ease of exposition, we first develop the theory under a degenerate
degree distribution $w_k=\mathbf{1}[k=k_0]$, that is, a fixed
deterministic degree $k$, and only later generalise to a non-trivial
degree distribution. To simplify the notation, we suppress the
$k$ in $(i,k)$ and write $i$ for an individual's state, with the
implicit understanding that $i$ is coupled with $k$.

\subsection{Formulating Probabilities Along a Branch}

We derive the probability density of a sampled tree from two fundamental quantities, following the backward differential-equation pruning scheme of \citet{maddison2007estimating}. Let $E_i(t)$ denote the probability that a
type-$i$ individual and all its descendants remain unobserved over an
elapsed interval of length $t$. Let $D^i_j(t)$ denote the
probability density that a clade rooted by a type-$i$ individual at
time $0$ produces exactly one sampled descendant of type $j$ at time
$t$, with no other observations in $(0,t]$. Throughout, we set the
present time to $t=0$ and let $t \ge 0$ denote the elapsed time
measured into the past, so that $(0,t]$ refers to relative time
rather than chronological or calendar time. We derive $E_i(t)$
first, and then use it in the derivation of $D^i_j(t)$.

\begin{figure}
\centering
\begin{tikzpicture}[> =stealth, node distance=1.5cm, every node/.style={scale=1}]

  \node[] (n0) {$i$};
  \node[below=3.0cm of n0] (j0) {};

  \draw[-,thick,black,dashed] (n0) -- (j0);

  \begin{scope}[xshift=3cm]
  \node[] (n1) {$i$};
  \node[below=3.0cm of n1] (j1) {};

  \draw[thick, dotted, -] (n1) -- (j1);
  \end{scope}

  \begin{scope}[xshift=7cm]
    \node (i3) {$i$};
    \coordinate (b3) at ([yshift=-1.2cm]i3.south);
    \draw[thick,-]   (i3.south) -- (b3);
    \node[left=of b3]  (ip1_3) {$i+1$};
    \node[right=of b3] (z3)    {$0$};
    \node[below=of ip1_3] (e3) {};
    \node[below=of z3]    (j3) {};

    \draw[dashed, thick, ->] (b3)   -- (ip1_3);
    \draw[dashed, thick, ->]  (b3)   -- (z3);
    \draw[dashed, thick]    (ip1_3) -- (e3);
    \draw[dashed, thick]  (z3)   -- (j3);
  \end{scope}

\end{tikzpicture}
\caption{Possible events for $E_{i}(t)$. Left: No infection, lineage recovers unobserved. Middle: No infection, no recovery, lineage remains unobserved. Right: Infection of a new lineage $0$, parent lineage $i+1$ and newborn lineage go unobserved.}
\label{fig:E_i}
\end{figure}
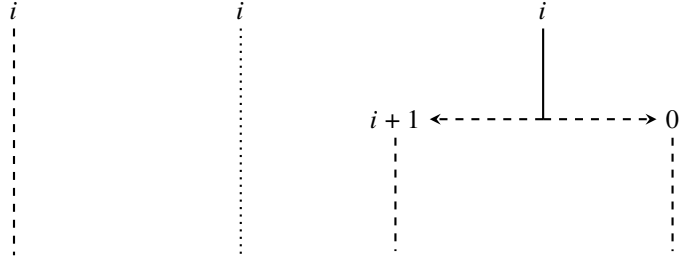

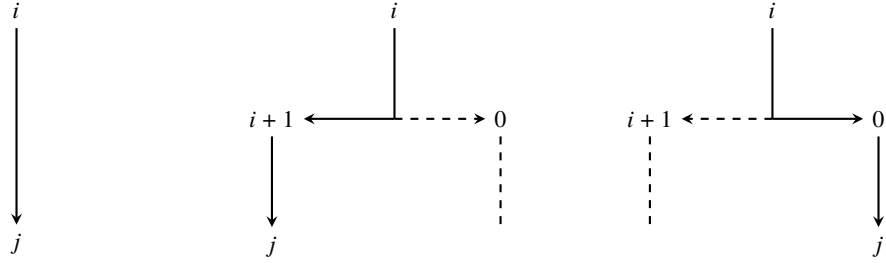
\begin{figure}
\centering
\begin{tikzpicture}[
    node distance = 1.2 cm and 1.2cm,
    every node/.style={font=\small},
    >=stealth
]
  \begin{scope}
    \node (i1) {$i$};
    \node[below=2.6cm of i1] (j_g1) {$j$};
    \draw[thick,->] (i1) -- (j_g1);
  \end{scope}

  \begin{scope}[xshift=5cm]
    \node (i2) {$i$};
    \coordinate (b2) at ([yshift=-1.2cm]i2.south);
    \draw[thick,-]   (i2.south) -- (b2);
    \node[left=of b2]  (ip1_2) {$i+1$};
    \node[right=of b2] (z2)    {$0$};
    \node[below=of ip1_2] (j2) {$j$};
    \node[below=of z2]    (e2) {};

    \draw[thick,->]  (b2)   -- (ip1_2);
    \draw[dashed, thick, ->] (b2)   -- (z2);
    \draw[thick,->]  (ip1_2) -- (j2);
    \draw[dashed, thick]    (z2)   -- (e2);
  \end{scope}

  \begin{scope}[xshift=10cm]
    \node (i3) {$i$};
    \coordinate (b3) at ([yshift=-1.2cm]i3.south);
    \draw[thick,-]   (i3.south) -- (b3);
    \node[left=of b3]  (ip1_3) {$i+1$};
    \node[right=of b3] (z3)    {$0$};
    \node[below=of ip1_3] (e3) {};
    \node[below=of z3]    (j3) {$j$};

    \draw[dashed, thick, ->] (b3)   -- (ip1_3);
    \draw[thick,->]  (b3)   -- (z3);
    \draw[dashed, thick]    (ip1_3) -- (e3);
    \draw[thick,->]  (z3)   -- (j3);
  \end{scope}

\end{tikzpicture}
\caption{Possible events for $D_{j}^{i}(t)$. Left: No infection, lineage recovers observed via survival. Middle: Infection of a new lineage $0$, the parent lineage (now in $i+1$) produced the sampled node at present, while the newborn lineage must go unobserved. Right: Infection of a new lineage $0$, the newborn lineage produced the sampled node at present, while the parent lineage (now in $+1$) must go unobserved.}
\label{fig:D_i_j}
\end{figure}

\subsubsection{Probability to Stay unobserved, \texorpdfstring{$E_{i}(t)$}{Ei(t)}}

To solve for $D_{j}^{i}(t)$, we first calculate the probability $E_{i}(t)$ that a type $i$ individual and all its descendants remain unobserved after a time interval $t$.

\begin{prop}
Let $i=0, \ldots, k$. The probability $E_{i}(t)$ that individual $i$ leaves no sampled descendants after time $t$ reads
\begin{eqnarray}
\frac{d}{d t} E_{i}(t)=\mu-\big(\gamma+ (k-i)\,\beta \big) E_{i}(t)+ (k-i)\,\beta E_{0}(t) E_{i+1}(t), \qquad E_i(0) =1.
\label{eq:E_fixed}
\end{eqnarray}
\end{prop}

\begin{proof}
With the present at $t=0$ and time measured backwards considering all possible events (see figure~\ref{fig:E_i}), we proceed recursively in the number of remaining susceptible contacts.

\paragraph{Base case ($i=k$)} When the focal individual has no downstream susceptible contacts left ($i=k$), the only events on $[0,t]$ are observed recovery (which would produce a sampled tip, and is therefore forbidden in the event of no sampled descendants) or unobserved recovery. Let $q(t)$ denote the probability of remaining infectious at time $t$, and $r(t)$ the probability of recovering unobserved by time $t$. These evolve as
\begin{eqnarray*}
\frac{d}{dt}q(t) = -\gamma q(t), \qquad \frac{d}{dt}r(t) = \mu q(t),
\end{eqnarray*}
with $q(0)=1$ and $r(0)=0$. Solving gives $q(t)=e^{-\gamma t}$ and $r(t)=\frac{\mu}{\gamma}(1-e^{-\gamma t})$. The total probability of remaining unsampled is thus
\begin{eqnarray*}
E_k(t) = q(t)+r(t) = \frac{\mu}{\gamma} + \frac{\sigma}{\gamma}e^{-\gamma t},
\end{eqnarray*}
which satisfies $\frac{d}{dt}E_k(t) = \mu - \gamma E_k(t)$ with $E_k(0)=1$.

\paragraph{Recursive step ($i<k$)} For an individual with $k-i$ remaining susceptible contacts, three events can occur in a small interval: no event, unobserved recovery at rate $\mu$, or infection at rate $(k-i)\beta$. The probability of remaining infectious, $q(t)$, now decays at the total rate $\gamma+(k-i)\beta$ of all events, while unobserved recovery accumulates at rate $\mu$. This yields
\begin{eqnarray*}
q(t) = e^{-(\gamma+(k-i)\beta)t}, \qquad r(t) = \frac{\mu}{\gamma+(k-i)\beta}\bigl(1-e^{-(\gamma+(k-i)\beta)t}\bigr).
\end{eqnarray*}
If an infection occurs at time $a\in(0,t]$, which happens with probability density $(k-i)\beta\, q(a)$, the parent lineage advances to state $i+1$ and the newborn enters state $0$. For the entire clade to remain unobserved, both must leave no sampled descendants over the remaining interval $t-a$, giving the product $E_0(t-a)E_{i+1}(t-a)$. Integrating over all possible infection times and adding the no-infection contribution,
\begin{eqnarray*}
E_i(t) = \frac{\mu}{\gamma+(k-i)\beta} + \frac{\sigma+(k-i)\beta}{\gamma+(k-i)\beta}e^{-(\gamma+(k-i)\beta)t} + \int_0^t (k-i)\beta\,e^{-(\gamma+(k-i)\beta)(t-a)}E_0(a)E_{i+1}(a)\,da.
\end{eqnarray*}
Differentiating in $t$ eliminates the integral and yields \eqref{eq:E_fixed}, with the convention $E_{k+1}\equiv 0$.
\end{proof}

\subsubsection{Probability density to produce a single sampled tip,
\texorpdfstring{$D^{i}_j(t)$}{Dij(t)}}

Let $0<\tau\le t$ be fixed, we now derive $D^{i}_j(t)$, the probability density that a clade rooted at time $0$ with type $i$ produces exactly one observation, of type $j$ at time $\tau$, and no other observation in $(0,t]$. By construction $D^{i}_j(t)$ depends on $t$ as well as on the sampling time $\tau$, which we suppress in the notation; as a density it is normalised with respect to
$\tau$, and $\tau$ enters only through the initial condition
$D^{i}_j(\tau)=\sigma\,\chi_i(j)$, which records that at $\tau$ the lineage
is sampled (rate $\sigma$) and is observed as type $i=j$. The forward ODE
in $t$ then propagates this initial density to any later time.

\begin{prop}
Let $i=0, \ldots, k$, the probability $D_{j}^{i}(t)$ that an individual $i$ evolves to a type \( j \) individual by the present time $t$ reads

\begin{equation}
\frac{d}{dt}D_{j}^{i}(t) = -\big(\gamma  + (k-i)\,\beta \big)\,D_{j}^{i}(t)+ (k-i)\,\beta\,E_{0}(t)\,D_{j}^{i+1}(t) + (k-i)\,\beta\,E_{i+1}(t)\,D_{j}^{0}(t),
\label{eq:D_fixed}
\end{equation}

where $E_{k+1}(t)=0$ and $D_{j}^{k+1}(t)=0$, with initial condition

\[
    D_{j}^{i}(\tau)=
\begin{cases}
    \sigma,& \text{if } j = i\\
    0,& \text{if } j \neq i
\end{cases}
\]

\end{prop}

\begin{proof}
At the sampling time $\tau$ the lineage is in state $i=j$ and is removed observed, giving the initial condition $D^i_j(\tau)=\sigma\,\chi_i(j)$ where $\chi_i(j)$ is the indicator function. For $t>\tau$ we proceed recursively in the number of remaining susceptible contacts $k-i$, tracing the clade backwards by considering all possible events (see figure~\ref{fig:D_i_j}).

\paragraph{Base case ($i=k$)} When the focal individual has no downstream susceptible contacts left ($i=k$), no infection events can occur. The only possibility is that the lineage survives from $\tau$ to $t$ without being observed, decaying at rate $\gamma$. Let $\phi(t)$ denote the probability density of this survival. Then
\begin{eqnarray*}
\frac{d}{dt}\phi(t) = -\gamma\,\phi(t), \qquad \phi(\tau)=\sigma,
\end{eqnarray*}
which solves to $\phi(t)=\sigma\,e^{-\gamma(t-\tau)}$. The probability density that the clade produces exactly one sampled tip of type $j$ is therefore
\begin{eqnarray*}
D^k_j(t) = \sigma\,e^{-\gamma(t-\tau)}\chi_k(j),
\end{eqnarray*}
satisfying $\frac{d}{dt}D^k_j(t) = -\gamma D^k_j(t)$ with the stated initial condition.

\paragraph{Recursive step ($i<k$)} For an individual with $k-i$ remaining susceptible contacts, three mutually exclusive events can occur on $[\tau,t]$: (a) no infection event, with the lineage simply surviving at total rate $\gamma+(k-i)\beta$; (b) an infection occurs at time $c\in(\tau,t)$, the parent advances to state $i+1$ and produces the sampled tip, while the newborn in state $0$ and all its descendants remain unobserved; (c) the same infection event occurs, but the newborn in state $0$ produces the sampled tip while the parent in state $i+1$ and all its descendants remain unobserved.

For event (a), the survival probability density decays as before, now at the combined rate of recovery and infection:
\begin{eqnarray*}
\phi(t) = \sigma\,e^{-(\gamma+(k-i)\beta)(t-\tau)}.
\end{eqnarray*}

For events (b) and (c), an infection occurs at time $c$ with probability density $(k-i)\beta\,e^{-(\gamma+(k-i)\beta)(t-\tau-c)}$. In event (b), the parent lineage (now in state $i+1$) must produce the sampled tip, giving factor $D^{i+1}_j(c)$, while the newborn clade (in state $0$) must remain unobserved, giving factor $E_0(c)$. In event (c), the roles are reversed: the parent and its descendants remain unobserved ($E_{i+1}(c)$), while the newborn clade produces the sampled tip ($D^0_j(c)$). Integrating over all possible infection times and summing the three contributions,
\begin{eqnarray*}
D^i_j(t) &=& \sigma\,e^{-(\gamma+(k-i)\beta)(t-\tau)}\chi_i(j) \\
&& + \int_0^{t-\tau}(k-i)\beta\,e^{-(\gamma+(k-i)\beta)(t-\tau-c)}\bigl[E_0(c)\,D^{i+1}_j(c) + E_{i+1}(c)\,D^0_j(c)\bigr]\,dc.
\end{eqnarray*}
Differentiating in $t$ eliminates the integral and yields \eqref{eq:D_fixed}, with the boundary conditions $E_{k+1}(t)=0$ and $D^{k+1}_j(t)=0$.
\end{proof}

\subsection{Generalisation to Random Degree}
\label{sec:random-degree}

Up to this point, our derivations assumed that each individual has a fixed contact degree $k$. This assumption allowed us to index types solely by the number of already infected contacts $i$, writing $E_i(t) \equiv E_{(i,k)}(t)$ and $D^i_j(t) \equiv D^{(i,k)}_{(j,k)}(t)$ with $k$ implicit. In real populations, however, individuals may differ in their total number of effective contacts.

As introduced in the previous section, we now assume that an individual's degree is random. We aim to understand how to adapt equations~\ref{eq:E_fixed} and~\ref{eq:D_fixed} in this case. Thereto, we again assume that the $i$ individual has a fixed degree $k$. However, we understand that we do not know the degree of any downstream individual. We introduce two newborn mixing terms $\widehat{E}_0(t)$ and $\widehat{D}_j^0$ to be, respectively,
the probability that a newly infected individual (conditioning
on the type, but not on the degree) will be unobserved for a time interval
$[0, t]$ together with the clade it might produce during this time interval,
and the probability density that the clade (regardless of the degree
of the newly infected downstream) at time $0$ produces at time $t$ an observation of
type $j$, while no other observation was triggered in the time interval
$[0, t]$.

Analytically, we condition on a latent degree $k$ for the focal lineage so that rates such as $(k-i)\beta$ remain explicit and the ODEs stay closed. The degree of any newborn lineage is unknown at infection and is integrated out with the distribution of $K$.

\begin{prop}
\label{prop:E-random}
Let $E_{(i,k)}(t)$ denote the probability that a type $(i,k)$ individual and all its descendants remain unobserved up to time $t$ (with $E_{(i,k)}(0)=1$). Define the degree-averaged newborn extinction probability
\[
\widehat{E}_0(t) \;=\; \sum_{\ell \ge 0} w_\ell\, E_{0,\ell}(t),
\quad\text{and more generally}\quad
\widehat{E}_i(t) \;=\; \sum_{\ell \ge i} \mathbb{P}(K=\ell \mid K \ge i)\, E_{i,\ell}(t).
\]
Then the fixed degree system generalises to the random case
\begin{equation}
\frac{d}{dt} E_{(i,k)}(t) \;=\;
\mu - \bigl(\gamma  + (k-i)\beta \bigr) E_{(i,k)}(t)
+ (k-i)\beta \, \widehat{E}_0(t)\, E_{(i+1,k)}(t),
\quad i=0,\dots,k,\;\; E_{k+1,k}\equiv 0.
\label{eq:E_random}
\end{equation}
\end{prop}

\begin{proof}
In the fixed degree case, the infection term uses $E_{(0,k)}(t)$ because a newborn inherits the same degree $k$ as the parent. With random degree, the newborn's degree is distributed as $K'$ and unknown at infection; hence $E_{0,k}(t)$ is replaced by the mixture $\widehat{E}_0(t)$. All other terms (recovery/observation and the parent's own transition $i\to i+1$) concern the focal lineage treated conditionally on its latent $k$. Marginalisation over $k$ is deferred to the likelihood stage. \qedhere
\end{proof}

\begin{prop}
\label{prop:random-D}
For each $0<\tau\le t$, let $D^{(i,k)}_j(t)$ denote the probability density
that a clade rooted at time $0$ in state $(i,k)$ produces, at time $\tau$,
exactly one observation of type $j$ (the degree of the observed individual
is integrated out), with no other observation in $(0,t]$. Then, with the
degree-averaged newborn density
$\widehat{D}^{\,0}_j(t)=\sum_{\ell\ge 0}w_\ell\,D^{(0,\ell)}_j(t)$,
\begin{eqnarray}
\frac{d}{dt}D^{(i,k)}_j(t)
=-\bigl(\mu+\sigma+(k-i)\beta\bigr)D^{(i,k)}_j(t)
+(k-i)\beta\,\widehat{E}_0(t)\,D^{(i+1,k)}_j(t)
+(k-i)\beta\,E_{(i+1,k)}(t)\,\widehat{D}^{\,0}_j(t),
\label{eq:D_random}
\end{eqnarray}

for $i=0,\dots,k$, with $E_{k+1,k}\equiv 0$, $D^{(k+1,k)}_j\equiv 0$, and
initial condition $
D^{(i,k)}_j(\tau)\;=\;\sigma\,\chi_i(j).
$
\end{prop}
\begin{proof}
In the fixed degree case, the two infection pathways involve $E_{(0,k)}D^{(i+1,k)}_{(j,k)}$ (newborn unobserved; parent later produces the sample) and $E_{(i+1,k)}D^{(0,k)}_{(j,k)}$ (parent unobserved; newborn produces the sample). Under random degree, both newborn terms are degree-averaged because the newborn's degree is unknown at infection, so $E_{(0,k)}\mapsto \widehat{E}_0$ and $D^{(0,k)}_{(j,k)}\mapsto \widehat{D}^0_j$. The lineage $i$ remains conditional on its latent $k$, hence $E_{(i+1,k)}$ and $D^{(i+1,k)}_j$ remain conditional on $k$. \qedhere
\end{proof}

\begin{figure}
	\centering
		\includegraphics[width=6cm]{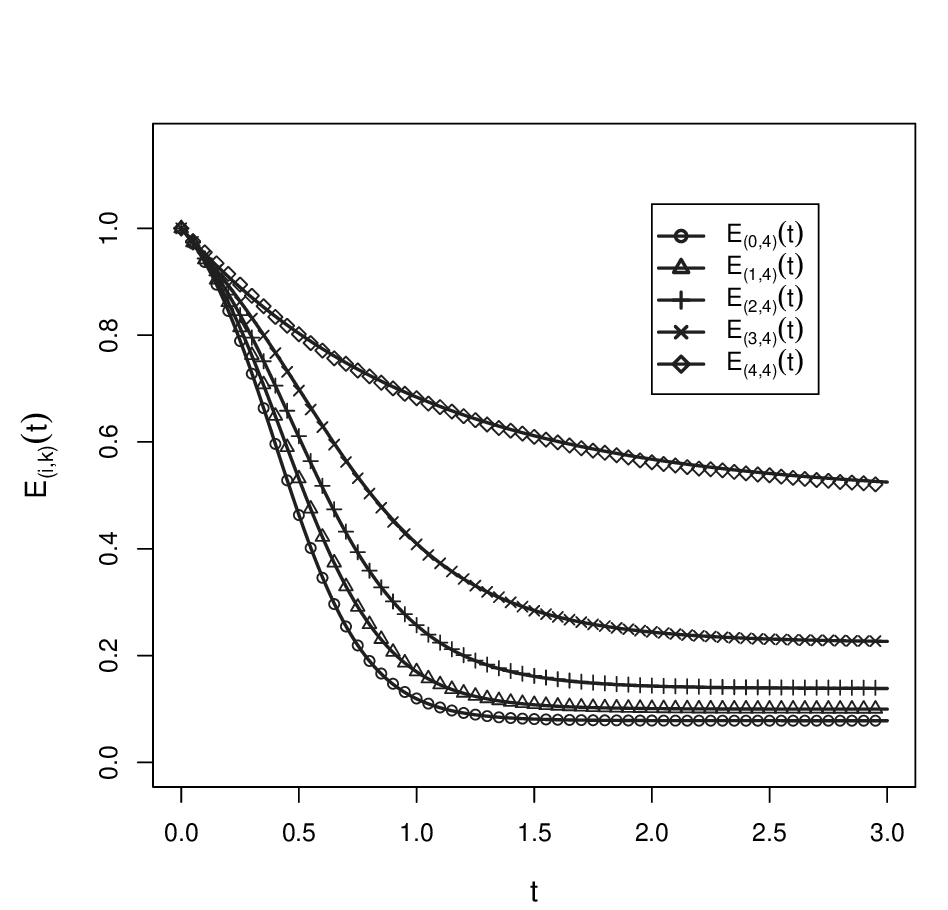}
	\caption{The probability $E_{(i,k)}(t)$ for a lineage and all of its descendants to go unobserved after time $t$. Gray thick lines: simulated trees for the different types, other points symbols for ODE results respectively - square:  $E_{0,k}(t)$, circle:  $E_{1,k}(t)$, triangle point up:  $E_{2,k}(t)$, plus: $E_{3,k}(t)$, cross:  $E_{k,k}(t)$. Choice of parameters: $\mu, \sigma = 0.5$, $\beta = 1.5$, $k = 4$ (fixed deterministic degree).}
	\label{fig:Eik-fixed}
\end{figure}

\begin{figure}
\centering
\includegraphics[width=0.95\textwidth]{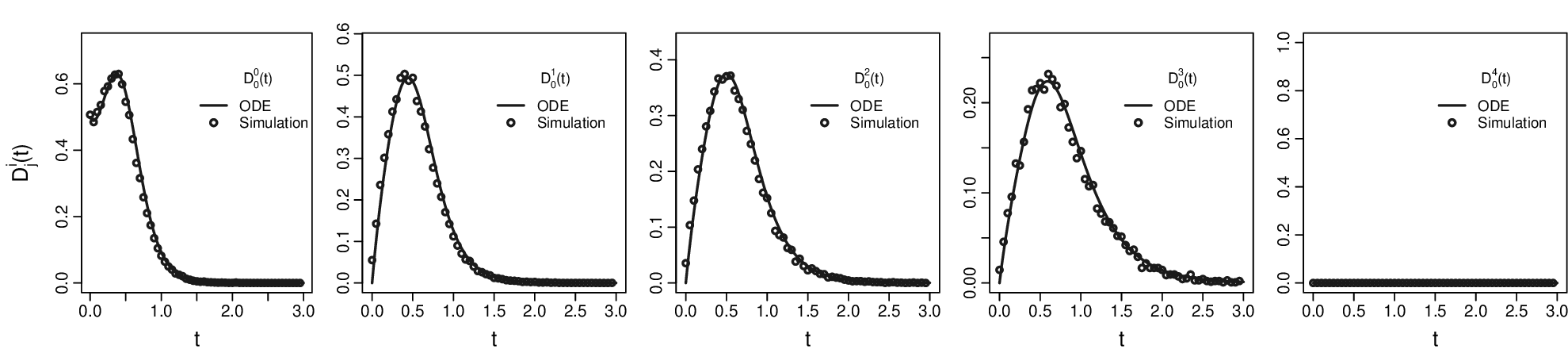}
\includegraphics[width=0.95\textwidth]{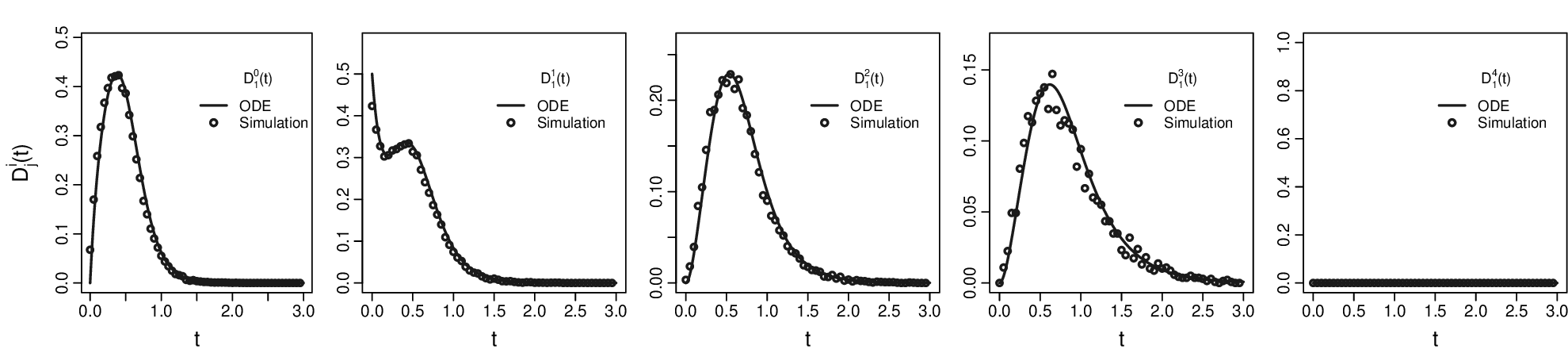}
\includegraphics[width=0.95\textwidth]{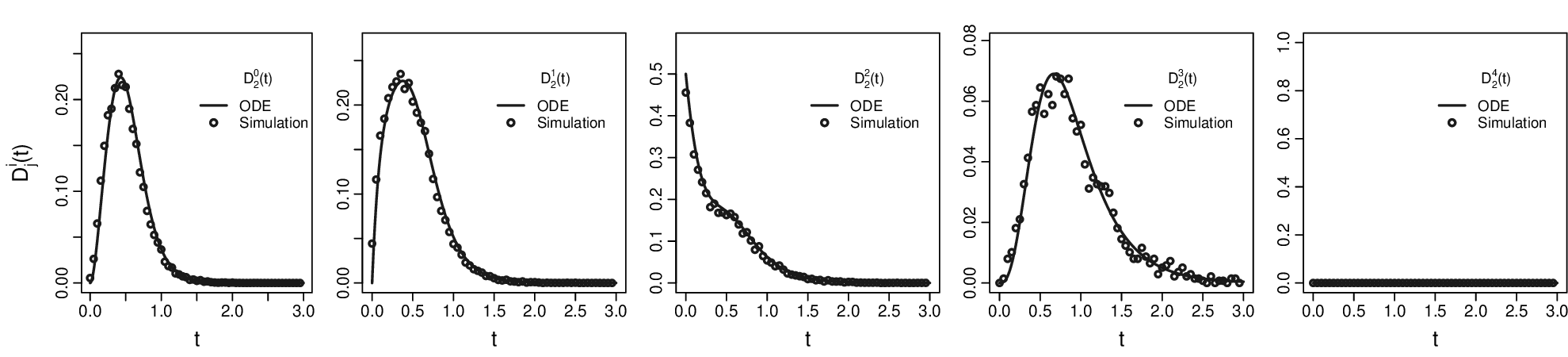}
\includegraphics[width=0.95\textwidth]{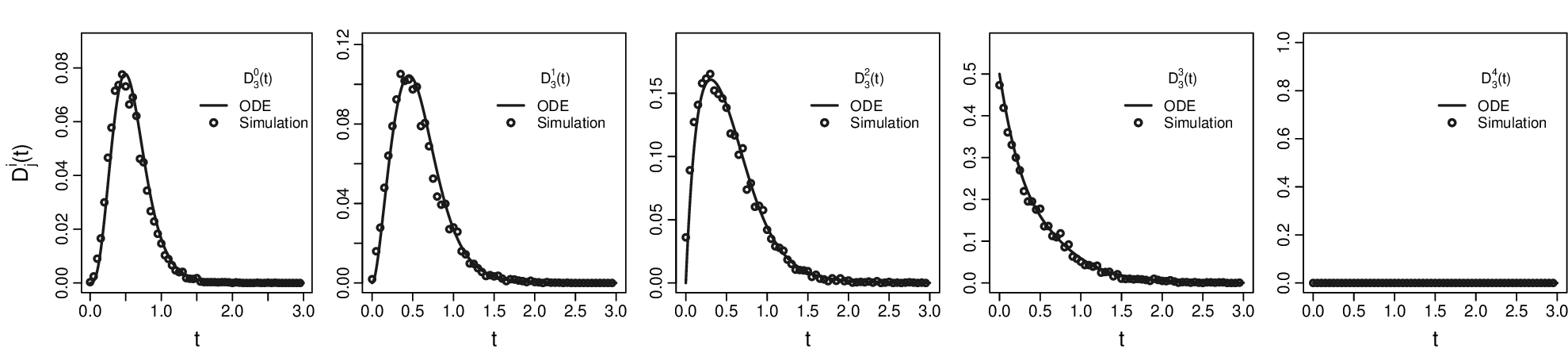}
\includegraphics[width=0.95\textwidth]{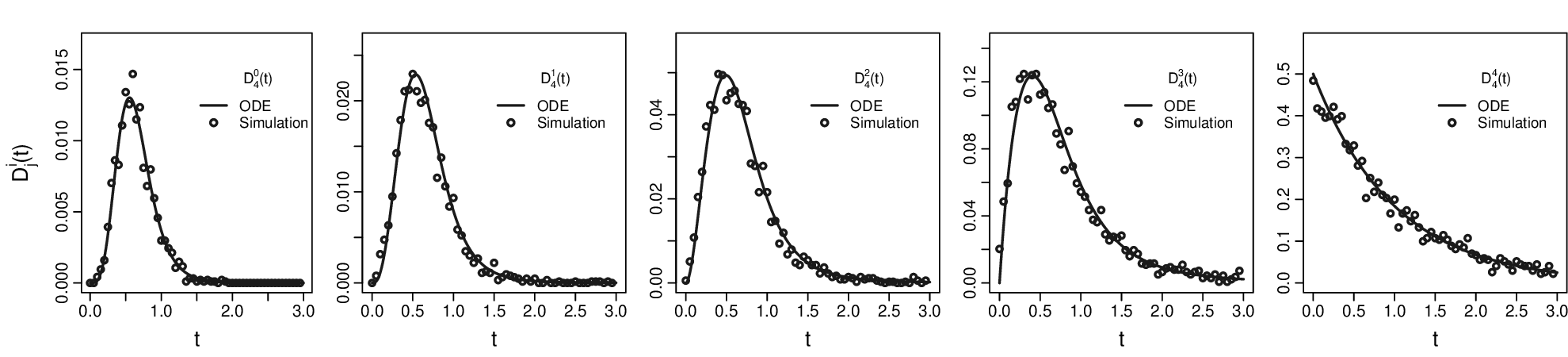}
\caption{The probability density $D_{j}^{(i,k)}(t)$ for a clade and all its descendants to have evolved as the observed sampled tree. Simulated trees (black circles), ODE theory results (black solid lines). Choice of parameters: $\mu, \sigma = 0.5$, $\beta = 1.5$, $k = 4$ (fixed deterministic degree).}
\label{fig:Dik-fixed}
\end{figure}

Figures~\ref{fig:Eik-fixed}--\ref{fig:Dik-fixed} compare the
analytical solutions of the survival probabilities $E_{(i,k)}(t)$ and the
sampling densities $D_{j}^{(i,k)}(t)$ to simulated trajectories under a fixed deterministic degree. Across all types $i$ and observed leaf types $j$, the theory agrees with simulations. The probability $E_{(i,k)}(t)$ that a clade remains unobserved depends strongly on the number of already infected downstream nodes
(Fig.~\ref{fig:Eik-fixed}). It is lowest at $i=0$, where all $k$
downstream edges are still available for transmission and any
descendant in the clade can trigger detection, and increases
monotonically with $i$ as the number of available transmission
pathways $(k-i)$ shrinks. At $i=k$, all downstream edges are
already infected; no further transmission is possible, giving the highest probability of going
unobserved.

For the saturated type $i=k$ of $D_{j}^{(i,k)}(t)$ (See last column of figure~\ref{fig:Dik-fixed}), the node has no remaining transmission
capacity, so the clade cannot grow and only the trajectory $j=k$
carries positive density, and $D_{j}^{(k,k)}(t) = 0$ for $j \neq k$, with the entire mass concentrated on $j=k$ corresponding to the evolution of the $i$ node itself. For the non-saturated type $i \in
\{0, 1, \dots, k-1\}$ (see the first four columns in figure~\ref{fig:Dik-fixed}), transmission events are still possible, and
$D_{j}^{(i,k)}(t)$ takes positive values for the full range of
observable leaf types $j$, in agreement with the simulated
trajectories.

\subsection{Stationary Distribution}
\label{subsection:stationary_distribution}

To evaluate the likelihood at a node of type $(i,k)$, we require the equilibrium frequencies $\pi_{i,k}$. A fundamental property of the continuous-time Markov process is that its stationary distribution is determined solely by the transition rates and not by which state a particular branch tip occupies.

\begin{prop}
\label{prop:stationary-random}
Let $n_{i,k}(t)$ denote the expected number of nodes of type $(i,k)$, $0\leq i\leq k$, at time $t$. The evolution is governed by the system of differential equations
\begin{eqnarray}
\frac{d}{d t} n_{0,k}(t) & = & -(\gamma + k\beta)\, n_{0,k}(t) \;+\; w_k\,\beta\, \sum_{k'} \sum_{j=0}^{k'-1}(k'-j)\, n_{j,k'}(t), \nonumber\\
\frac{d}{d t} n_{i,k}(t) & = & -\bigl(\gamma + (k-i)\beta\bigr)\, n_{i,k}(t) \;+\; (k-i+1)\,\beta\, n_{i-1,k}(t), \quad i\geq 1.
\label{eq:stationary_distribution_ode_random}
\end{eqnarray}
In the exponentially growing phase with rate $r$, the equilibrium frequencies satisfy, for each degree class $k$,
\[
\pi_{i,k} \;=\; \frac{(k-i+1)\beta}{\gamma + (k-i)\beta + r}\,\pi_{i-1,k}, \qquad i\geq 1,
\]
which can be expressed in product form
\[
\pi_{i,k} \;=\; \pi_{0,k}\,\prod_{j=1}^{i}\frac{(k-j+1)\beta}{\gamma + (k-j)\beta + r}, \qquad i=1,\ldots,k,
\]
with degree-conditional frequencies $\pi_{i\mid k}=\pi_{i,k}\big/\sum_{j=0}^{k}\pi_{j,k}$ obeying the same recursion. The growth rate $r$ is the unique positive solution of the scalar eigenvalue equation
\begin{equation}
1 \;=\; \beta\, \sum_{k} \frac{w_k}{\gamma + k\beta + r}\left[\,k \;+\; \sum_{j=1}^{k-1}(k-j) \prod_{i=1}^{j} \frac{(k-i+1)\beta}{\gamma + (k-i)\beta + r}\,\right].
\label{eq:rnd-eigen-main}
\end{equation}

\begin{proof}
See section 1 in the supplementary material.
\end{proof}
\end{prop}

\begin{prop}
\label{prop:spectral-gap-random}
Let $\mathbf{n}(t)$ solve the system~\eqref{eq:stationary_distribution_ode_random} on the joint type space $\{(i,k):0\leq i\leq k\}$, with generator matrix $A$, and let $r$ be its dominant eigenvalue. Denote the equilibrium distribution by $\boldsymbol{\pi}=(\pi_{i,k})$ and define the spectral gap
\[
\delta \;=\; r - \max_{\ell \geq 1} \Re\lambda_\ell \;>\; 0,
\]
where $\{\lambda_\ell\}$ are the eigenvalues of $A$. Then for some constant $C>0$,
\[
\bigl\|\mathbf{p}(t) - \boldsymbol{\pi}\bigr\| \;\leq\; C\,e^{-\delta t},
\qquad
\mathbf{p}(t) \;=\; \frac{\mathbf{n}(t)}{\sum_{k'}\sum_{j=0}^{k'} n_{j,k'}(t)},
\]
so that the joint frequency vector $\mathbf{p}(t)$ converges exponentially fast to $\boldsymbol{\pi}$ at rate $\delta$, independently of the initial condition.
\end{prop}
\noindent Proof. See section 2 in the supplementary material.

\begin{cor}
\label{cor:Spi-random}
Let $S_\pi(k) \;=\; \sum_{i=0}^{k}(k-i)\,\pi_{i\mid k}$ denote the equilibrium-expected number of remaining susceptible downstream neighbours conditional on degree $k$, and let
\[
\mathbb{E}[S_\pi] \;=\; \sum_{k} w_k\, S_\pi(k) \;=\; \mathbb{E}[K] - \mathbb{E}[I],
\]
where $K\sim w_k$ is the random degree and $I$ counts already infected downstream neighbours.
The first limit below is taken within a fixed degree class as $k\to\infty$, and the second is the population average under a degree distribution $w_k$ whose mass concentrates on large degrees,
\[
S_\pi(k) \;\xrightarrow{\,k\to\infty\,}\; k-1, \qquad
\mathbb{E}[S_\pi] \;\approx\; \mathbb{E}[K]-1.
\]
This reproduces the classical SIR on tree behaviour, such that on average, one downstream neighbour of a typical lineage is already infected.

\begin{proof}
See section 3 in the supplementary material.
\end{proof}
\end{cor}

\subsection{Basic Reproductive Number}\label{subsec:r0}

For a contact network with random degree $K\sim w_k$ and per-edge transmission rate $\beta$, an infectious individual of degree $k$ has $k$ transmission opportunities and a mean infectious period $1/\gamma$. The standard per-node reproduction number conditional on degree is
\begin{equation}
R_0(k) \;=\; \frac{k\,\beta}{\gamma},
\qquad
R_0 \;=\; \mathbb{E}[R_0(K)] \;=\; \frac{\mathbb{E}[K]\,\beta}{\gamma}.
\label{eq:basic_r0_random}
\end{equation}

\paragraph{Type-dependent reproduction number}
For an infectious individual currently in state $(i,k)$, the $i$ already infected downstream neighbours represent secondary cases that have already occurred; they are part of the realised, not the future, offspring. The remaining $k-i$ downstream edges still lead to susceptible contacts, each transmitting at rate $\beta$ over a residual mean infectious period $1/\gamma$. We therefore define the type-dependent (residual) reproduction number
\begin{equation}
R_0(i,k) \;=\; \frac{(k-i)\,\beta}{\gamma},
\qquad i=0,1,\ldots,k,
\label{eq:type_r0_random}
\end{equation}
which counts only further secondary cases produced after the current type $(i,k)$. A naive count that includes the $i$ realised secondary cases would yield $i+(k-i)\beta/\gamma$; we drop the $i$ term because $R_0(i,k)$ is intended as a forward-looking quantity describing transmission potential going forward, not a tally of realised infections. In particular, $R_0(0,k)=k\beta/\gamma=R_0(k)$, recovering the per-degree basic reproduction number.

\paragraph{Equilibrium-weighted residual reproduction number}
In the exponentially growing phase, the types $(i,k)$ are populated according to the equilibrium distribution with degree-conditional frequencies $\pi_{i\mid k}=\pi_{i,k}/\sum_{j=0}^{k}\pi_{j,k}$, $\sum_{i=0}^{k}\pi_{i\mid k}=1$. Weighting the type-dependent residual reproduction number by these equilibrium frequencies gives the equilibrium-weighted residual reproduction number of a degree-$k$ lineage,
\begin{equation*}
\bar{R}_0(k)\;=\;\sum_{i=0}^{k}\pi_{i\mid k}\,R_0(i,k)
\;=\;\frac{\beta}{\gamma}\sum_{i=0}^{k}\pi_{i\mid k}(k-i)
\;=\;\frac{R_0(k)}{k}\,S_\pi(k).
\label{eq:Rbar0-k-random}
\end{equation*}
Averaging over the random degree gives the population-level equilibrium-weighted residual reproduction number
\begin{equation}
\bar{R}_0 \;=\; \mathbb{E}\bigl[\bar{R}_0(K)\bigr] \;=\; \sum_{k} w_k\,\bar{R}_0(k)
\;=\;\frac{\beta}{\gamma}\,\mathbb{E}[S_\pi].
\label{eq:Rbar0-random}
\end{equation}
$\bar{R}_0(k)$ counts only the remaining (future) secondary cases expected from a lineage of degree $k$ that has already infected $\mathbb{E}[I\mid k]=k-S_\pi(k)$ of its contacts. These quantities are the natural population-level summaries of realised transmission potential once the type distribution has equilibrated: because they discount the downstream contacts a lineage has already infected, they lie below the basic reproduction number $R_0$ and measure the forward transmission that remains, a hierarchy we return to when interpreting the Karnataka estimates in Section~\ref{sec:empirical-karnataka}.

\section{Likelihood}
\label{sec:likelihood}

\subsection{Construction of the tree likelihood}\label{sec:construction}

We now derive the likelihood of an observed sampled tree $\mathcal{T}$
under the structured transmission model introduced above. Each infected individual carries the augmented character type $(i,k)$, where $k$ is the downstream degree and $i$
records the number of already infected downstream contacts. When a
transmission event occurs from a lineage in state $(i,k)$, two daughter
lineages are produced with distinct roles. The continuing lineage is the same individual and, advances from $(i,k)$ to $(i+1,k)$, that is,
its degree is conserved. The newborn lineage is the newly infected
individual and always starts in state $(0,K')$ where $K'$ is an independent
draw from the degree distribution $w_k$.

Let $\hat{t}_0<0$ denote the root time, with time measured from past to
present and the present fixed at $t=0$. The sampled tree has $n\geq 1$
observed tips at times $\hat{t}_0\le t_1\le\cdots\le t_n\le 0$, with observed character type $j_\ell\in\{0,\dots,k\}$ at tip $\ell$; the degree of
each sampled individual is not observed. The tree contains $n-1$ internal
branching nodes at times $\tau_1,\dots,\tau_{n-1}$. The root state is
$(i_0,k)$ with $i_0$ unknown and $k$ the (latent) degree. These absolute times serve only to order the events of the tree and to define the elapsed length of each edge. Every $E$ and $D$ entering the likelihood is evaluated at such an elapsed length, a non-negative difference of two event times, which is exactly the relative time $t \ge 0$ of Section 2.1. The root, for instance, contributes through $E_{(i_0,k)}(-\hat t_0)$, where $-\hat t_0 > 0$ is the elapsed time from the root to the present.

\subsubsection{Edges of two kinds, governed by a single ODE}
\label{sec:two-edge-types}

Every edge of the sampled tree terminates at exactly one of an observation (a sampled tip) or a transmission (an internal
branching node) event. We describe the contribution of an edge by a single per-edge density
$D^{(i,k)}_{j}(t)$ that solves the ODE of Proposition~\ref{prop:random-D}
along the edge. The two edge types share this ODE and differ only in the
initial condition imposed at the lower (i.e.\ more recent) end of the
edge.

\paragraph{Tip-ending edge}
The edge starts in state $(i,k)$ at time $u$ and ends at a sampled tip in
observed state $j$ at time $v>u$. The terminating event is an observation,
which occurs at the instantaneous sampling rate $\sigma$ when the lineage
is in the observed state. The contribution of the edge to the likelihood
is
\[
D^{(i,k)}_{j}(v-u),
\qquad\text{with initial condition}\qquad
D^{(i,k)}_{j}(0)=\sigma\,\chi_i(j),
\]
where the elapsed time is measured from the tip backwards to the start of
the edge.

\paragraph{Branching-ending edge}
The edge starts in state $(i,k)$ at time $u$ and ends at an internal
branching event at time $v>u$, with the lineage in state $(j,k)$ just
before the branch. The terminating event is a transmission, which occurs
at rate $(k-j)\beta$. Along the edge itself the per-edge density also satisfies
the same ODE of Proposition~\ref{prop:random-D}, but the initial condition at
the lower end is no longer the simple $\sigma\,\chi_i(j)$ of a tip.
Instead, it is supplied recursively by the resulting subtrees attached at the branch, multiplied by the branching rate $(k-j)\beta$.

\subsubsection{Elementary building blocks of the likelihood}
\label{sec:building-blocks}

The likelihood of the sampled tree $\mathcal{T}$ factorises into a product
of edge contributions and internal node contributions, which we now
collect.

\paragraph{Terminal edge from a continuing lineage}
If the edge starts at $(i,k)$ at time $u$ and ends at a sampled tip $j$ at time $v>u$, its contribution is
\[
D^{(i,k)}_{j}(v-u).
\]

\paragraph{Terminal edge from a newborn lineage}
If the edge originates as a newborn at a branching event (with random
degree $K'$ integrated out) and ends at a sampled tip $j$ after elapsed time $v-u$, its contribution is
\[
\widehat{D}^{\,0}_{j}(v-u).
\]

\paragraph{Internal branching node}
Consider an internal branching event at time $\tau$ for a parent lineage
of type $(i,k)$. Let $\mathcal{T}_A$ and $\mathcal{T}_B$ denote the two
daughter subtrees emanating from this node, with elapsed times $\Delta_A$
and $\Delta_B$ along their respective root edges to the next event. At
the branch, the continuing daughter enters $(i+1,k)$ and the
newborn daughter enters $(0,K')$.
Because the sampled tree records only tip observations and does not
identify which daughter produced which subtree, both role assignments
must be summed. Writing $D^{(i+1,k)}_{\mathcal{T}_X}(\Delta_X)$ for the
likelihood contribution of subtree $\mathcal{T}_X$ assuming its root
edge is the continuing daughter (conserved degree $k$, state $i+1$),
and $\widehat{D}^{\,0}_{\mathcal{T}_X}(\Delta_X)$ for the corresponding
contribution assuming the root edge is the newborn daughter (degree
integrated out, type $0$), the branching node contributes
\begin{equation}\label{eq:branching-contribution}
(k-i)\beta\;\Bigl[\;
D^{(i+1,k)}_{\mathcal{T}_A}(\Delta_A)\,\widehat{D}^{\,0}_{\mathcal{T}_B}(\Delta_B)
\;+\;
D^{(i+1,k)}_{\mathcal{T}_B}(\Delta_B)\,\widehat{D}^{\,0}_{\mathcal{T}_A}(\Delta_A)
\;\Bigr].
\end{equation}
The first term corresponds to the role assignment in which the continuing
daughter carries $\mathcal{T}_A$ and the newborn daughter carries
$\mathcal{T}_B$; the second term is the reverse assignment. The notation $D^{(i+1,k)}_{\mathcal{T}_X}(\Delta_X)$ should be read
recursively. When $\mathcal{T}_X$ is a single observed tip $j$, this reduces to the terminal-edge contribution
$D^{(i+1,k)}_{j}(\Delta_X)$ given above; when $\mathcal{T}_X$ contains
further internal branching nodes, the contribution is built up from the
same set of building blocks applied to the substructure of $\mathcal{T}_X$.
The same recursive interpretation applies to
$\widehat{D}^{\,0}_{\mathcal{T}_X}(\Delta_X)$.

These local contributions assemble into the full tree likelihood.
The two trivial cases are immediate. In case (1), if the entire clade descending from the root leaves no sampled descendant by the present, $L(\mathscr{T}\mid\theta)=E_{(i_0,k)}(-\hat t_0)$. In case (2), if the tree contains a single sampled tip of type $j_1$ at time $t_1$, $L(\mathscr{T}\mid\theta)=D^{(i_0,k)}_{j_1}(t_1-\hat t_0)$. We now turn to case (3), a nontrivial worked example that exposes the general pattern.

\subsubsection*{Worked example: three observations}
We consider the tree in Figure~\ref{fig:tree3obs} with $n=3$ observed tips, two internal branching nodes and fixed topology.
Let the tips be observed as types $j_1,j_2,j_3$ at times
$t_1,t_2,t_3$ respectively, and let the branching times be $\tau_1$ and $\tau_2$,
with root at $\hat t_0$ and root degree $k$. We define the elapsed times
\[
\Delta_0=\tau_2-\hat t_0,\qquad
\Delta_1=\tau_1-\tau_2,\qquad
\Delta_{12}=t_1-\tau_2,\qquad
\Delta_{21}=t_2-\tau_1,\qquad
\Delta_{31}=t_3-\tau_1.
\]

\begin{figure}[ht]
\centering
\begin{tikzpicture}[
    scale=0.9,
    transform shape,
    >=Stealth,
    every node/.style={font=\footnotesize},
    tip/.style      = {circle, fill=black, inner sep=2.5pt},
    bnode/.style    = {circle, draw=black, fill=white,
                       line width=1.0pt, inner sep=2.5pt},
    rootnode/.style = {rectangle, fill=black, inner sep=3.0pt},
    cont/.style     = {thick, black},
    newb/.style     = {thick, dashed, red!70!black},
    taxs/.style     = {thin, black!50}
]

\node[tip, label=below:{$t_1,\;j_1$}] (t1) at (1.0, 0.0) {};
\node[tip, label=below:{$t_3,\;j_3$}] (t3) at (3.5, 0.0) {};
\node[tip, label=below:{$t_2,\;j_2$}] (t2) at (6.0, 0.0) {};

\node[bnode] (tau1) at (4.75, 3.0) {};
\node[bnode] (tau2) at (4.75, 5.2) {};

\node[rootnode,
      label={[font=\footnotesize]right:{$(i_0,k)$}}]
      (root) at (4.75, 7.0) {};

\draw[cont] (root) -- (tau2)
    node[left, midway, xshift=-3pt, font=\footnotesize]
        {$D^{(i_0,k)}_{i_0+1}(\Delta_0)$};

\draw[cont]
    (tau2) -- (tau1);
\node[left, font=\footnotesize, xshift=-3pt] at ($(tau2)!0.5!(tau1)$)
    {$D^{(i_0+1,k)}_{i_0+2}(\Delta_1)$};

\draw[newb]
    (tau2) -- (tau2 -| t1)
    -- (t1);
\node[right, font=\footnotesize, xshift=3pt]
    at ($(tau2 -| t1)!0.5!(t1)$)
    {$\widehat{D}^{0}_{j_1}(\Delta_{12})$};

\draw[cont]
    (tau1) -- (tau1 -| t2)
    -- (t2);
\node[right, font=\footnotesize, xshift=3pt]
    at ($(tau1 -| t2)!0.5!(t2)$)
    {$D^{(i_0+2,k)}_{j_2}(\Delta_{21})$};

\draw[newb]
    (tau1) -- (tau1 -| t3)
    -- (t3);
\node[left, font=\footnotesize, xshift=-3pt]
    at ($(tau1 -| t3)!0.5!(t3)$)
    {$\widehat{D}^{0}_{j_3}(\Delta_{31})$};

\node[right=6pt of tau2, red!70!black, font=\footnotesize]
    {$(k{-}i_0)\beta$};
\node[right=30pt of tau1, red!70!black, font=\footnotesize]
    {$(k{-}i_0{-}1)\beta$};

\draw[black!70, ->, line width=0.8pt] (9.5, 7.4) -- (9.5, -0.4)
    node[below, black!80, font=\footnotesize]{$t=0$ (present)};  
\node[black!80, font=\footnotesize, rotate=90] at (9.95, 3.5)  
    {past $\longrightarrow$};

\draw[black!50, thin, line width=0.6pt] (9.3,0.0) -- (9.5,0.0);
\node[left, black!80, font=\scriptsize] at (9.25,0.0) {$t_3,\;t_2,\;t_1$};
\draw[black!50, thin, line width=0.6pt] (9.3,3.0) -- (9.5,3.0);
\node[left, black!80, font=\scriptsize] at (9.25,3.0) {$\tau_1$};
\draw[black!50, thin, line width=0.6pt] (9.3,5.2) -- (9.5,5.2);
\node[left, black!80, font=\scriptsize] at (9.25,5.2) {$\tau_2$};
\draw[black!50, thin, line width=0.6pt] (9.3,7.0) -- (9.5,7.0);
\node[left, black!80, font=\scriptsize] at (9.25,7.0) {$\hat{t}_0$};

\begin{scope}[shift={(0.9,-2.2)}]
  \draw[cont] (0,0) -- (1.3,0)
      node[right, font=\footnotesize]
          {continuing lineage};
  \draw[newb] (0,-0.7) -- (1.3,-0.7)
      node[right, font=\footnotesize]
          {newborn lineage};
  \node[tip,      label=right:{\footnotesize sampled tip}]     at (4.8, 0.0){};
  \node[bnode,    label=right:{\footnotesize infection event}] at (4.8,-0.7){};
  \node[rootnode, label=right:{\footnotesize root}]            at (8.0,0.0){};
\end{scope}
\end{tikzpicture}
\caption{Sampled tree for Case~(3) with $n=3$ observations: time runs vertically downward.
  The root individual ($(i_0,k)$, time $\hat{t}_0$) runs as a single vertical stem until its first infection at $\tau_2$, which
  occurs at rate $(k-i_0)\beta$ and advances it to $(i_0+1,k)$.
  A newborn (type $0$, unknown degree, dashed) branches left toward tip $t_1$;
  its edge uses $\widehat{D}^0_{j_1}$ because the newborn's degree is integrated out.
  The continuing lineage (type $(i_0+1,k)$) runs to $\tau_1$, where a second
  infection at rate $(k-i_0-1)\beta$ advances it to $(i_0+2,k)$;
  another newborn (state $0$) branches toward tip $t_3$ using $\widehat{D}^0_{j_3}$,
  while the continuing lineage reaches tip $t_2$ with $D^{(i_0+2,k)}_{j_2}$.
  }
\label{fig:tree3obs}
\end{figure}
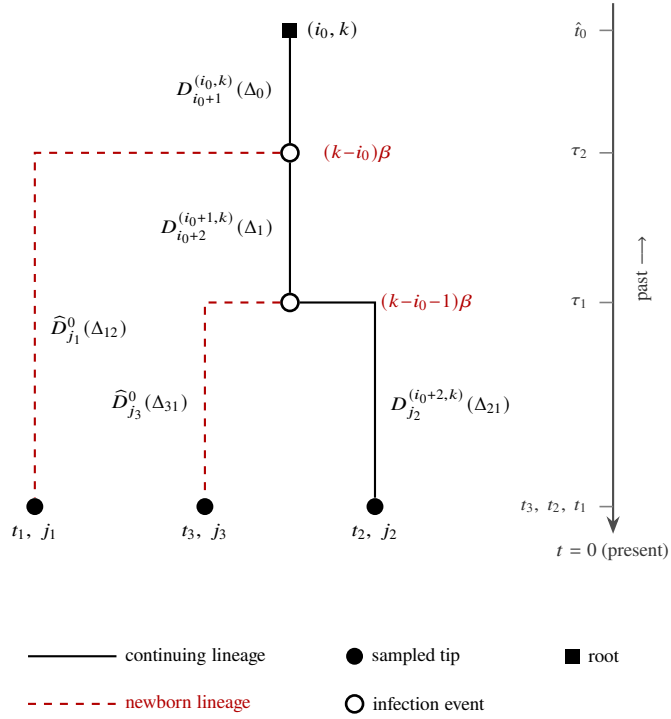

At the more recent branching event $\tau_1$, the parent lineage is of type $(i_0+1,k)$, so the infection occurs at rate
$(k-i_0-1)\beta$. The continuing daughter enters $(i_0+2,k)$ and
the newborn daughter enters $0$ with unknown degree. The contribution, summing over both role assignments, is
\begin{equation}
\label{eq:L4-final}
L_4
=
(k-i_0-1)\beta
\left[
D^{(i_0+2,k)}_{j_2}(\Delta_{21})\,\widehat{D}^{0}_{j_3}(\Delta_{31})
+
D^{(i_0+2,k)}_{j_3}(\Delta_{31})\,\widehat{D}^{0}_{j_2}(\Delta_{21})
\right].
\end{equation}

At the older branching event $\tau_2$, the parent lineage is of type $(i_0,k)$, so the infection occurs at rate $(k-i_0)\beta$. The continuing
daughter enters $(i_0+1,k)$ and the newborn daughter enters
$0$ with unknown degree. One daughter leads to the subtree rooted at $\tau_1$, while the
other leads to tip $(t_1,j_1)$. Summing over both assignments gives
\begin{equation}
\label{eq:L5-final}
L_5
=
(k-i_0)\beta
\left[
D^{(i_0+1,k)}_{\,i_0+2}(\Delta_1)\,\widehat{D}^{0}_{j_1}(\Delta_{12})
+
\widehat{D}^{0}_{\,i_0+2}(\Delta_1)\,D^{(i_0+1,k)}_{j_1}(\Delta_{12})
\right] L_4.
\end{equation}
Note that in the second term, the subtree $A$ rooted at $\tau_1$ is carried by the newborn daughter, accordingly $D^{(i_0+1,k)}_{i_0+2}(\Delta_1)$ becomes $\widehat{D}^{0}_{i_0+2}(\Delta_1)$ because that daughter has unknown degree. The continuing daughter instead leads to the single tip $j_1$, contributing $D^{(i_0+1,k)}_{j_1}(\Delta_{12})$.

Finally, the root edge from $\hat t_0$ to $\tau_2$ carries the root lineage from $i_0$ to $i_0+1$ over elapsed time
$\Delta_0$. Its contribution is
\begin{equation}
\label{eq:L6-final}
L_6
=
D^{(i_0,k)}_{\,i_0+1}(\Delta_0)\,L_5.
\end{equation}
Substituting \eqref{eq:L4-final} and \eqref{eq:L5-final} into
\eqref{eq:L6-final} yields the complete likelihood for the three-tip tree
\begin{align}
\label{eq:lik3-final}
L(\mathscr{T}\mid\theta,k)
=\;&
D^{(i_0,k)}_{\,i_0+1}(\Delta_0)\,
(k-i_0)\beta\,
(k-i_0-1)\beta
\notag\\
& \times
\left[
D^{(i_0+1,k)}_{\,i_0+2}(\Delta_1)\,\widehat{D}^{0}_{j_1}(\Delta_{12})
+
\widehat{D}^{0}_{\,i_0+2}(\Delta_1)\,D^{(i_0+1,k)}_{j_1}(\Delta_{12})
\right]
\notag\\
&\qquad \times
\left[
D^{(i_0+2,k)}_{j_2}(\Delta_{21})\,\widehat{D}^{0}_{j_3}(\Delta_{31})
+
D^{(i_0+2,k)}_{j_3}(\Delta_{31})\,\widehat{D}^{0}_{j_2}(\Delta_{21})
\right].
\end{align}
The structure of equation~\eqref{eq:lik3-final} makes the general pattern transparent. Each edge of the continuing lineage contributes a $D_*^{(i,k)}$ factor retaining the latent degree $k$. Each edge carrying a newborn lineage contributes a $\widehat{D}_*^{0}$ factor in which the degree has been integrated out according to $w_k$. Each internal node contributes an instantaneous rate $(k-i)\beta$, and both role assignments at each branching event are summed.

\subsection{Likelihood Conditioned on Non-Trivial Sampling}
\label{sec:theorem-likelihood}

The argument above extends to arbitrary $n$ and any rooted binary
time-stamped topology. In practice we know a priori that at least one individual was
sampled ($n \ge 1$) and we condition on this event $\{n > 0\}$.  By definition, $E_{(i_0,k)}(-\hat t_0)$ is the probability that a lineage of type $(i_0,k)$ at the root leaves no sampled descendant by the present, so $P(n > 0 \mid i_0,k) = 1-E_{(i_0,k)}(-\hat t_0)$.
The root degree $k$ is latent and $i_0$ is unknown; marginalising jointly over $(i_0,k)$ using the degree weights $w_k$ and the degree-conditional equilibrium weights $\pi_{i_0|k}$ gives the following result.

\begin{theorem}
\label{thm:likelihood}
Given a sampled tree $\mathscr{T}$ with $n \ge 1$ observations,
conditioned on $n > 0$ and on the root time $\hat{t}_0$, the
likelihood under random contact degree is
\begin{eqnarray}
    \label{eq:likelihood-conditioned}
    L\!\left(\mathscr{T}\mid\theta,\,\hat{t}_0,\,n{>}0\right)
    \\ \nonumber
    \;&=&\;
    \sum_{k \ge 0} w_k
    \left[\sum_{i_0=0}^{k}
          \frac{\pi_{i_0 \mid k}}{1 - E_{(i_0,k)}(-\hat{t}_0)}\right]
    \cdot
    D^{(i_0,k)}_{i_{0 + 1}}(\Delta_0)
    \cdot
    \prod_{\ell=1}^{n} D^{(i_\ell,k)}_{j_\ell}(\Delta_\ell) \nonumber
     \\
   && \qquad \cdot \prod_{m=1}^{n-1}
    (k-i_m)\beta
    \left[
        D^{(i_m+1,k)}_{\mathcal{T}_{A_m}}(\Delta^{+}_{m}) \, \widehat{D}^{0}_{\mathcal{T}_{B_m}}(\Delta^{-}_{m})
        +
        D^{(i_m+1,k)}_{\mathcal{T}_{B_m}}(\Delta^{-}_{m}) \, \widehat{D}^{0}_{\mathcal{T}_{A_m}}(\Delta^{+}_{m})
    \right],\nonumber
\end{eqnarray}
where $k$ is the root degree, integrated out with weights $w_k$; $i_0 \in \{0, \ldots, k\}$ is the state of the root lineage, marginalised with the degree-conditional equilibrium weights $\pi_{i_0 \mid k}$; $i_{0 + 1}$ is the state reached by the root lineage at the first branching node $\tau_1$; $\Delta_0 = \tau_1 - \hat{t}_0$ is the elapsed time from root to the first branching node; $i_\ell$ is the continuing lineage at the origin of edge $\ell$; $j_\ell \in \{0, \ldots, k\}$ is the observed type at tip $\ell$; $\Delta_\ell$ is the elapsed time along edge $\ell$; $i_m$ is the continuing lineage immediately before branching event $m$; $\mathcal{T}_{A_m}, \mathcal{T}_{B_m}$ are the types or subtree labels of the two daughters at node $\tau_m$; $\Delta^{+}_{m}$ and $\Delta^{-}_{m}$ are the elapsed times along the two outgoing daughter edges; $\widehat{D}^0_{*}$ are the degree-averaged newborn observation densities; $\pi_{i_0 \mid k}$ are the stationary weights conditional on degree $k$; and $1 - E_{(i_0,k)}(-\hat{t}_0)$ is the probability that a root lineage $(i_0,k)$ produces at least one sampled descendant.

\end{theorem}

\begin{proof}
The proof follows exactly the construction of Section~\ref{sec:construction}, with two modifications reflecting random degree. First, the root degree $k$ is latent and enters as an outer summation weighted by $w_k$; all quantities depending on the root lineage \big($D^{(i,k)}$, $E_{(i,k)}$, rates $(k-i)\beta$\big) are conditional on $k$. Second, at every internal branching node, the newborn daughter's degree is unknown at infection and integrated out, that is, $D^{(0,k)}_*$ is replaced by $\widehat{D}^0_*$. The continuing daughter retains degree $k$ throughout. The root conditioning, the product structure over edges and branching nodes, and the sum over the two daughter role assignments proceed accordingly.
\end{proof}

\begin{remark}
\label{rem:topology}
The likelihood depends on the sampled tree only through (i) the observed
tip times $\{t_\ell\}$ and types $\{j_\ell\}$; (ii) the internal branching
node times $\{\tau_m\}$ and types $\{i_m\}$; and (iii) the rooted binary
topology. Once these are fixed, the likelihood is obtained entirely from
elapsed edge lengths $\Delta_0, \Delta_\ell, \Delta^{\pm}_{m}$, infection
rates $(k-i_m)\beta$ and the sum over the two daughter role assignments
at each branching node. The only structural difference from the fixed degree
case is that newborn daughters use $\widehat{D}_*^0$ rather than $D_*^{(0,k)}$.
\end{remark}

\subsection{Epidemiological reparametrisation}
\label{sec:reparam}

Working directly with $(\beta,\mu,\sigma)$ is statistically inconvenient because
the likelihood surface contains a near flat ridge along the direction of
simultaneous scaling of all three rates. We replace $(\beta,\mu,\sigma)$ by three epidemiological
interpretable quantities

\begin{equation*}
  \label{eq:reparam_inv}
  \beta \;=\; \frac{R_0(k)\,\gamma}{k},
  \qquad
  \sigma \;=\; p_{\mathrm{obs}}\,\gamma,
  \qquad
  \mu \;=\; (1 - p_{\mathrm{obs}})\,\gamma.
\end{equation*}
Here $R_0(k)$ is the basic reproduction number of a newly infected individual with
degree $k$, $\gamma$ is the total removal
rate, and $p_{\mathrm{obs}}$ is the probability that removal is accompanied by sampling.
Under random degree, the unconditional basic reproduction number is the degree-weighted average
\begin{equation*}
  \label{eq:R0_random}
  R_0 \;=\; \mathbb{E}[K]\,\frac{\beta}{\gamma}
       \;:=\; \mu_K \,\frac{\beta}{\gamma},
\end{equation*}
where $\mu_K = \sum_k k\,w_k$ is the expected contact degree. The parameter vector becomes $(R_0, \mu_K, \gamma, p_{\mathrm{obs}}, w_k)$. Under this reparameterisation, the conditional unobserved and observed clade probability for a degree $k$ lineage satisfies
\begin{equation*}
  \label{eq:E_reparam}
  \frac{d}{dt}E_{(i,k)}
  \;=\;
  \gamma\!\left[
    (1-p_{\mathrm{obs}})
    - \!\left(1 + \frac{(k-i)R_0(k)}{k}\right)\!E_{(i,k)}
    + \frac{(k-i)R_0(k)}{k}\,\widehat{E}_0\,E_{(i+1,k)}
  \right],
\end{equation*}

and
\begin{equation*}
\frac{d}{dt}D^{(i,k)}_{j}
= -\,\gamma\!\left[\left(1 + \frac{(k-i)\,R_0(k)}{k}\right) D^{(i,k)}_{j}
+ \frac{(k-i)\,R_0(k)}{k}\,\widehat{E}_{0}\,D^{(i+1,k)}_{j} + \frac{(k-i)\,R_0(k)}{k}\,E_{(i+1,k)}\,\widehat{D}^{0}_{j}\right],
\label{eq:Dij_reparam}
\end{equation*}

respectively  where $\widehat{E}_0 = \sum_\ell w_\ell E_{0,\ell}$ and $\widehat{D}^{0}_{j}(t) = \sum_{\ell\ge 0} w_{\ell}\,D^{(0,\ell)}_{j}(t)$ are evaluated across the degree distribution with initial condition $D^{(i,k)}_{j}(\tau) = p_{\mathrm{obs}}\,\gamma\,\chi_{i}(j)$. The factor $\gamma$ appears multiplicatively throughout

\begin{equation}
  \label{eq:scaling}
  E_{(i,k)}(t;\,R_0(k),\gamma,p_{\mathrm{obs}})
  \;=\;
  \tilde{E}_{(i,k)}\!\left(\gamma t;\,R_0(k),p_{\mathrm{obs}}\right),
  \quad
  D^{(i,k)}_j(t)
  \;=\;
  \gamma\,\tilde{D}^{(i,k)}_j\!\left(\gamma t;\,R_0(k),p_{\mathrm{obs}}\right),
\end{equation}
where $\tilde{E}_{(i,k)}$ and $\tilde{D}^{(i,k)}_j$ are dimensionless functions of $\gamma t$ alone. The extra factor of $\gamma$ in $D_j^{(i, k)}$ arises because $D_j^{(i, k)}$ is a density with respect to time. The solution $D^{(i,k)}_{j}$, jointly with the $E_{(i,k)}$
system, shows that the likelihood depends on $\gamma$ only through the time rescaling
$t \mapsto \gamma t$, the per-degree reproduction number $R_{0}(k)$, and the global Jacobian factor $\gamma^{n}$ from the $n$ tip initial conditions $\sigma\,\chi_{i}(j) = p_{\mathrm{obs}}\,\gamma\,\chi_{i}(j)$.
The scaling relation~\eqref{eq:scaling} separates the roles of the three parameters. First, $R_0(k)$ and the degree distribution $w_k$ determine the shape of the type distribution. Second, $\gamma$ sets the absolute timescale, in the sense that when branch lengths are expressed in calendar time, $\gamma$ is identifiable through the time rescaling; without an external clock, $\gamma$ must be fixed externally. Third, $p_{\text {obs }}$ governs the mix between unobserved and unobserved recovery (sampling); in practice it is the least precisely estimated parameter and is often fixed from external surveillance data.


\section{Application to Estimation: Simulation Study}
\label{sec:application_est_sim_study}

We now validate the likelihood framework on simulated trees under two
scenarios that differ in how much of the transmission phylogeny is
observable.
In the first, the sampled tree is resolved, that is, every internal
branching time, lineage state and sampled tip are known.
In the second, only the sampled tips are known, while a fraction of all internal
structure is latent, reflecting the situation encountered in practice
when sequenced cases are available, but the full transmission chain
cannot be reconstructed. For our estimator analysis, outbreak trees are generated using the stochastic simulator described in Section~4 of the supplementary material for a fixed deterministic degree such that $w_k = \mathbf{1}[k = k_0]$. Trees are retained only if they produce at least one sampled tip,
yielding $N = 753$ valid trees.

\subsection{Full Tree: Maximum Likelihood Estimation}
\label{subsec_mle_known_tree_top}

\subsubsection{Estimation setup}
\label{subsubsec:estimation_setup}

For each tree the log-likelihood $\ell_r(\theta) =
\sum_r \ell_r(\theta)$ decomposes into five additive terms: the
observed tip density $\log D^i_{j_{\mathrm{obs}}}(\delta)$, unobserved
sub-clade extinction $\log E_i(\delta)$, inter-branch survival
$-(\gamma + (k-i)\beta)\delta$, branching rate $\log(k-i)\beta$, and
root conditioning $\log\pi_{i_0} - \log(1-E_{i_0}(\Delta))$.
The functions $E_{(i,k)}$ and $D^{(i,k)}_j$ solve the backward Kolmogorov system
of Section~\ref{sec:model}, evaluated by linear
interpolation on a pre-computed ODE grid.

In our dataset, $4,596$ of the $9,630$ total edge rows are observed tips, giving an empirical fraction of $4,596/9,630 = 0.4773$. This is consistent with the true simulation value of $p_{obs}=0.5$, and we accordingly fix $p_{\mathrm{obs}} \approx 0.5$ throughout. We treat $R_0$ as the sole continuous free parameter and select $k$
by the Akaike Information Criterion (AIC), setting $\gamma = 1.0$.
Fixing $\gamma$ is necessary because scaling all branch lengths and
$\gamma$ by the same constant leaves the likelihood unchanged, making
$\gamma$ non-identifiable without an external clock.
Fixing $p_{\mathrm{obs}}$ is similarly required, as shown by
\citet{stadler2009incomplete}, the infection rate, recovery rate and sampling probability are not jointly identifiable from tree shape alone, and the
profile likelihood in $p_{\mathrm{obs}}$ is strictly monotone, so this parameter must be supplied from external surveillance data.
For each $k \in \{1,\ldots,12\}$, we minimise $-\ell(\theta)$ over
$R_0 \in (0.5, 30)$ and recover $\hat\beta = \hat{R}_0\gamma/k$.
Model selection uses $\mathrm{AIC} = 2n_p - 2\ell$ with $n_p = 2$
free parameters; profile-likelihood 95\% confidence intervals are
obtained by inverting the likelihood ratio.

\subsubsection{MLE results}
\label{subsubsec:agg_mle}

Table~\ref{tab:mle_full} reports the aggregate results.
The AIC-selected model is $k = 4$ (AIC weight $1.00$), yielding
$\hat{R}_0 = 6.04$ (95\% CI: $[5.84,\,6.24]$) and $\hat\beta = 1.51$
(95\% CI: $[1.46,\,1.56]$).
Both true values $R_0 = 6.0$ and $\beta = 1.5$ lie inside their
respective intervals, confirming accurate recovery under correct
model specification.

One practical and important feature is how $\hat{R}_0$ changes only slowly from $6.04$ at $k = 4$ to $4.77$ at $k=12$ (left panel figure~\ref{fig:mle_result}), a shift of just over one
unit as $k$ changes by eight.
Over the same range, $\hat\beta$ varies more than three-fold, from $1.51$
to $0.40$.
The product $k\hat\beta$ slow drift reflects a genuine partial non-identifiability in the sense that the tree constrains the total branching rate $k\beta$ rather than $k$ and $\beta$
individually, so the estimate of $R_0 = k\beta/\gamma$ is far more stable
than the estimates of either factor alone. The survival (inter-branch interval) terms of the log-likelihood drive
this behaviour, because a tree with $n$ sampled tips contains of order
$n\log n$ inter-branch segments but only $n$ tip segments.
Those survival terms take the form $
  \sum_{n} \bigl[-(\gamma + (k-i)\beta)\,\delta\bigr]$
and the MLE maximises this sum by finding the $\beta$ that makes the
observed inter-branch intervals $\delta$ look most typical under the hazard
rate $\gamma + (k-i)\beta$.
This condition approximately pins $k\hat\beta$ to the tree's branching
tempo, largely independent of the assumed $k$.
Fixing $k$ and estimating $\beta$ is therefore mathematically equivalent
to fixing $\beta$ and estimating $k$, so in both cases, the free parameter
adjusts until $k\hat\beta$ matches the data.
The AIC-selected estimate at $k = 4$ can thus serve as a reliable reference
point, and provides a natural hedge against uncertainty in $k$.

Additionally, the right panel in figure~\ref{fig:mle_result} confirms that the estimated stationary
distribution $\pi_{i}$ agrees closely with the empirical state
frequencies across all simulated trees. Numerically,
\[
(\pi_0,\pi_1,\pi_2,\pi_3,\pi_4)
\approx (0.4286,\,0.2857,\,0.1714,\,0.0857,\,0.0286),
\quad r\approx 3.5,
\]
so, at equilibrium, about \(43\%\) of individuals are of type $0$,
\(29\%\) type 1, \(17\%\) type 2, \(9\%\) type 3, and \(3\%\) type
4. The growth rate \(r\approx 3.5\) implies that, once the chain is
in its exponential phase, the total number of active individuals
grows by a factor of \(e^{3.5}\) per time unit.

\begin{figure}[ht]
    \centering
    \includegraphics[width=0.33\textwidth]{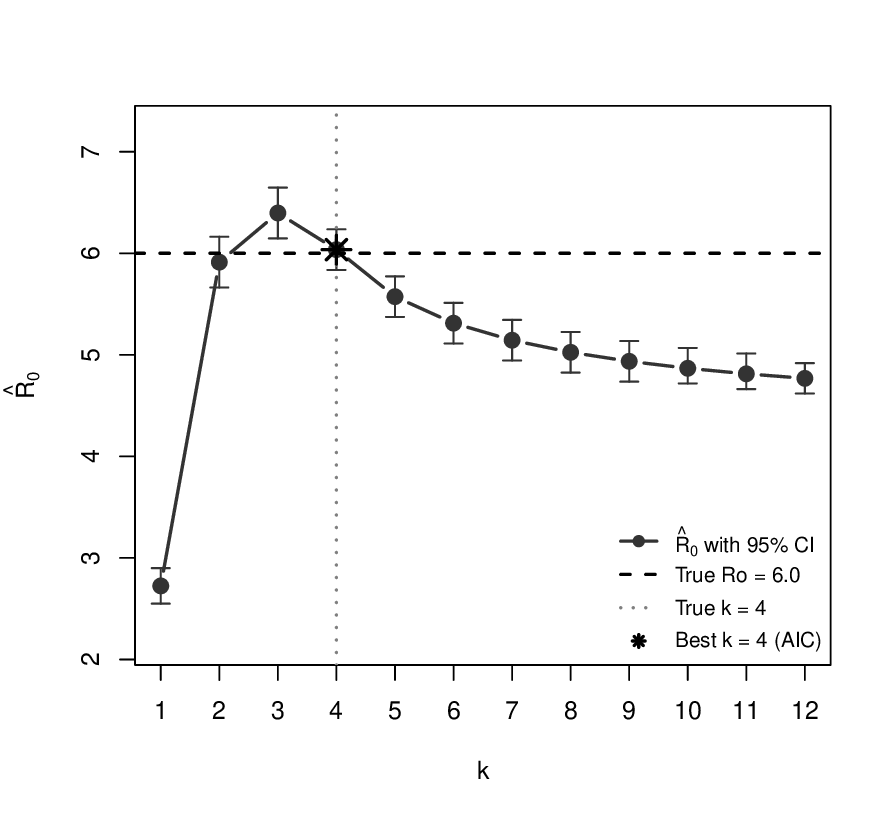}
    \includegraphics[width=0.33\textwidth]{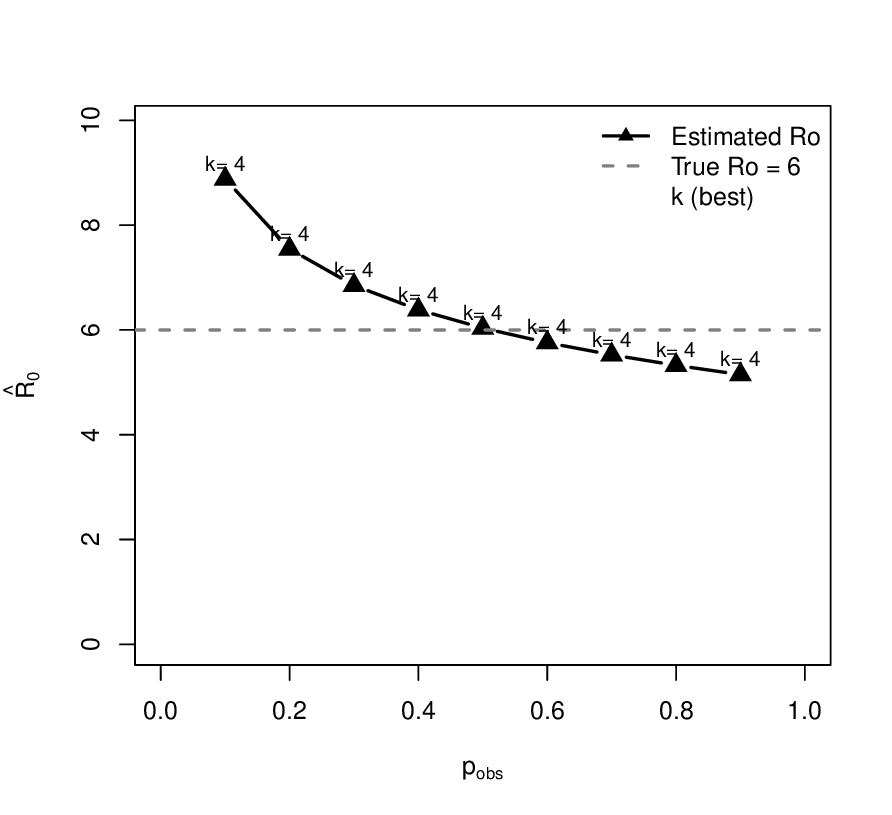}
    \includegraphics[width=0.33\textwidth]{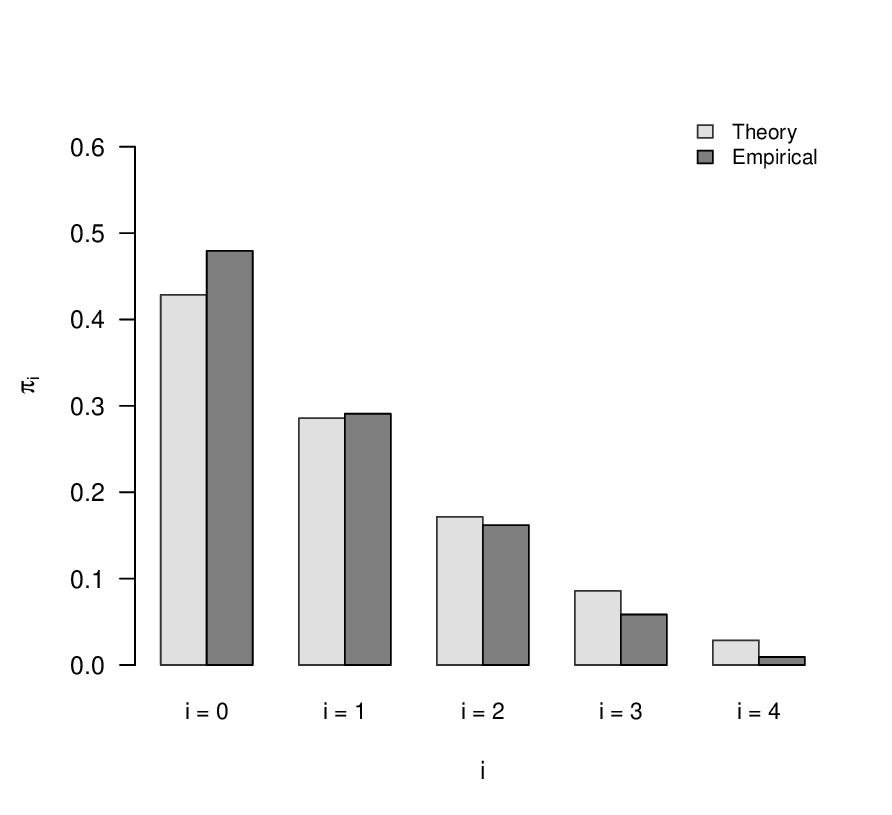}
    \caption{Left: $\hat{R}_0$ and $\hat\beta$ as functions of the
    assumed $k$ (true $k=4$ starred).
    Centre: $\hat{R}_0$ as a function of the assumed $p_{\mathrm{obs}}$
    (true value $0.5$ marked).
    Right: theoretical equilibrium frequencies $\pi_{i|k}$ versus
    empirical frequencies across 753 simulated trees. Choice of parameters:
    $\mu, \sigma = 0.5$, $\beta = 1.5$, $k = 4$ (fixed deterministic
    degree).}
    \label{fig:mle_result}
\end{figure}

\begin{table}[ht]
\small
\centering
\caption{Aggregate MLE results across $k = 1,\ldots,12$ from $N = 753$
simulated trees. Fixed: $\gamma = 1.0$, $p_{\mathrm{obs}} = 0.5$.
True values: $k = 4$ (starred), $R_0 = 6.0$, $\beta = 1.5$.
$\Delta\mathrm{AIC}$ is relative to the selected model.}
\label{tab:mle_full}
\begin{tabular}{crrrr}
\toprule
$k$ & $\hat{R}_0$ [95\% CI] & $\hat\beta$ [95\% CI] & $\Delta$AIC & AIC wt \\
\midrule
$4^{*}$ & 6.04 $[5.84,\;6.24]$ & 1.51 $[1.46,\;1.56]$ &     0.0 & 1.00 \\
5       & 5.57 $[5.37,\;5.77]$ & 1.11 $[1.08,\;1.16]$ &    60.1 & 0.00 \\
6       & 5.31 $[5.11,\;5.51]$ & 0.89 $[0.85,\;0.92]$ &   118.4 & 0.00 \\
7       & 5.14 $[4.94,\;5.34]$ & 0.73 $[0.71,\;0.76]$ &   164.1 & 0.00 \\
8       & 5.03 $[4.83,\;5.23]$ & 0.63 $[0.60,\;0.65]$ &   199.8 & 0.00 \\
9       & 4.94 $[4.74,\;5.14]$ & 0.55 $[0.53,\;0.57]$ &   228.2 & 0.00 \\
10      & 4.87 $[4.72,\;5.07]$ & 0.49 $[0.47,\;0.51]$ &   251.2 & 0.00 \\
11      & 4.81 $[4.66,\;5.01]$ & 0.44 $[0.42,\;0.46]$ &   270.1 & 0.00 \\
12      & 4.77 $[4.62,\;4.92]$ & 0.40 $[0.39,\;0.41]$ &   285.9 & 0.00 \\
\midrule
3 & 6.40 $[6.15,\;6.65]$ & 2.13 $[2.05,\;2.22]$ &     $3.16{\times}10^{5}$ & 0.00 \\
2 & 5.91 $[5.66,\;6.16]$ & 2.96 $[2.83,\;3.08]$ &     $1.20{\times}10^{6}$ & 0.00 \\
1 & 2.72 $[2.55,\;2.90]$ & 2.72 $[2.55,\;2.90]$ &     $3.16{\times}10^{6}$ & 0.00 \\
\bottomrule
\end{tabular}
\end{table}

\subsubsection{\texorpdfstring{Sensitivity to $p_{\mathrm{obs}}$}{Sensitivity to pobs}}
\label{sec:pobs_sensitivity}

The centre panel of Figure~\ref{fig:mle_result} shows $\hat{R}_0$ as
$p_{\mathrm{obs}}$ varies from $0.1$ to $0.9$ with all other settings
fixed.
AIC correctly selects $k = 4$ throughout the range, but $\hat{R}_0$
shifts from $\approx 8.88$ at $p_{\mathrm{obs}} = 0.1$ to
$\approx 5.15$ at $p_{\mathrm{obs}} = 0.9$, a total range of $3.73$
units around the true $R_0 = 6.0$.
The mechanism is intuitive in both directions. A high
$p_{\mathrm{obs}}$ means the estimator expects almost every removal to generate a sampled tip; each observed tip, therefore, contributes only modestly to the likelihood, and a slower transmission rate suffices to explain a modestly lower $R_0$ in the tree. Conversely, a low $p_{\mathrm{obs}}$ treats each sampled tip as a rare event, so the estimator must postulate very fast transmission
to account for the many observed lineages within the observation
window, consequently leading to a higher $R_0$ in the tree.
In either case, the NLL at the optimum is substantially larger than
under the correct $p_{\mathrm{obs}} = 0.5$, confirming that the
model fit is degraded rather than merely reparameterised.
These findings reinforce the non-identifiability result; errors of
even a moderate size in the assumed sampling probability can shift
$\hat{R}_0$ by more than a full unit, making reliable external
estimates of surveillance coverage an essential prerequisite for
valid inference.

\subsection{Partially Resolved Tree: Bayesian }
\label{subsec:bayesian_tips_only}

When the full phylogeny is incomplete or unresolved, only the sampled tips
$\mathcal{D} = \bigl\{x_i : i=1,\ldots,n\bigr\}$, with $x_i=(t_i,j_i)$, are mostly
observed, while some or all internal branching times and states are latent. We estimate $R_0$ and $k$ jointly via a Bayesian approach, and Metropolis-Hastings Monte Carlo Markov Chain (MCMC) methods, analytically marginalising over the unknown
internal states.

\subsubsection{Bayesian model and prior specification}
\label{subsubsec:bayes_model}

At each latent branching node with parent-edge duration $\delta$, the
fixed state survival and branching rate contributions are replaced by
the weighted sum
\begin{equation}
  \log\!\Bigl(
    \textstyle\sum_{s=0}^{k}
    \pi_s\,(k-i)\beta\,e^{-(\gamma+(k-i)\beta)\delta}
  \Bigr),
  \label{eq:state_marg}
\end{equation}
where $\pi_s$ is the equilibrium frequency of type $s$.
This marginalisation eliminates the latent states exactly while
preserving the likelihood structure. The joint posterior over $R_0$, $k$ and all latent branch time
vectors $\{\boldsymbol{\tau}^{(r)}\}$ is
\begin{equation}
  \pi\!\bigl(R_0, k, \{\boldsymbol{\tau}^{(r)}\} \mid \{\mathcal{D}^{(r)}\}\bigr)
  \;\propto\;
  \prod_{r=1}^{N} \tilde{L}\!\bigl(\mathcal{D}^{(r)} \mid R_0, k,
    \boldsymbol{\tau}^{(r)}\bigr)\cdot\pi(R_0)\cdot\pi(k)
  \cdot \prod_{m} \mathbf{1}\!\bigl[\tau_m^{(r)}\in[\ell_m^{(r)},u_m^{(r)}]\bigr],
  \label{eq:full_posterior}
\end{equation}
where $r$ indexes a valid sampled tree and each latent time $\tau_m$ is constrained to the feasibility interval $[\ell_m,u_m]$ defined by the surrounding observed events: $\ell_m$ is the start time of the parent lineage's current edge and $u_m$
is the time of the child lineage's next recorded event. Any $\tau_m$ outside $[\ell_m,u_m]$ produces a non-positive elapsed time on an incident edge and renders the likelihood undefined.
We place a $\mathrm{LogNormal}(\log 5,\,1)$ prior on $R_0$ (median 5,
95\% mass in $[0.74,\,33.8]$), a discrete uniform prior $k \sim \mathrm{Uniform}\{1,\ldots,K_{\max}\}$, and a uniform prior
on $\tau_m$ over $[\ell_m,u_m]$.

\begin{figure}[ht]
\centering
\includegraphics[width=0.35\linewidth]{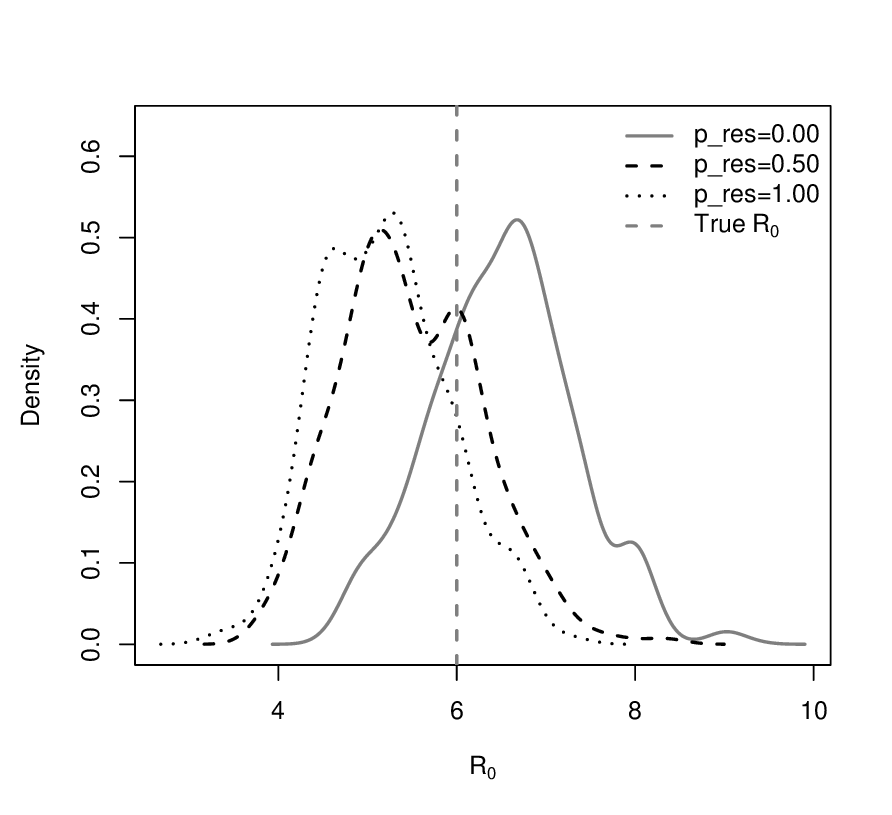}
\includegraphics[width=0.35\linewidth]{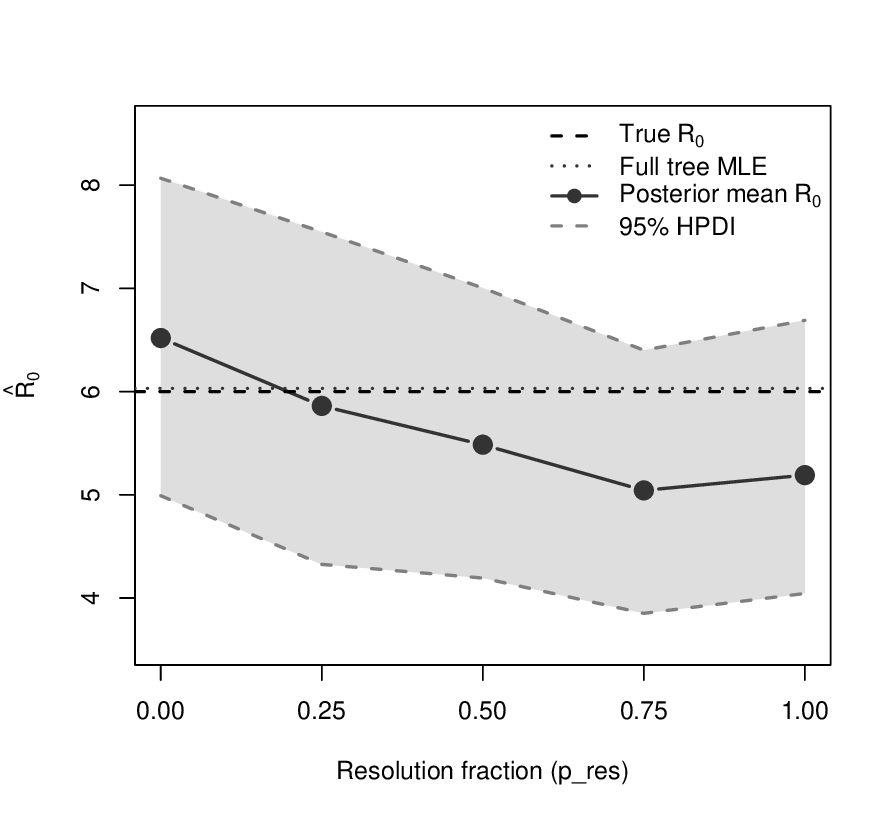}
\includegraphics[width=0.35\linewidth]{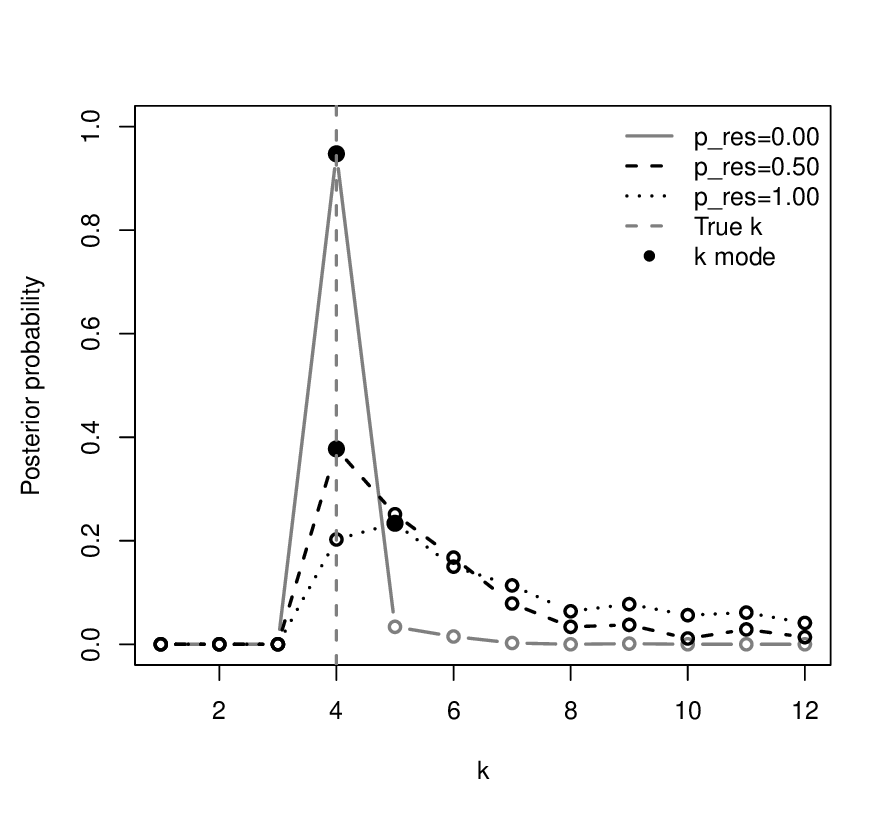}
\includegraphics[width=0.35\linewidth]{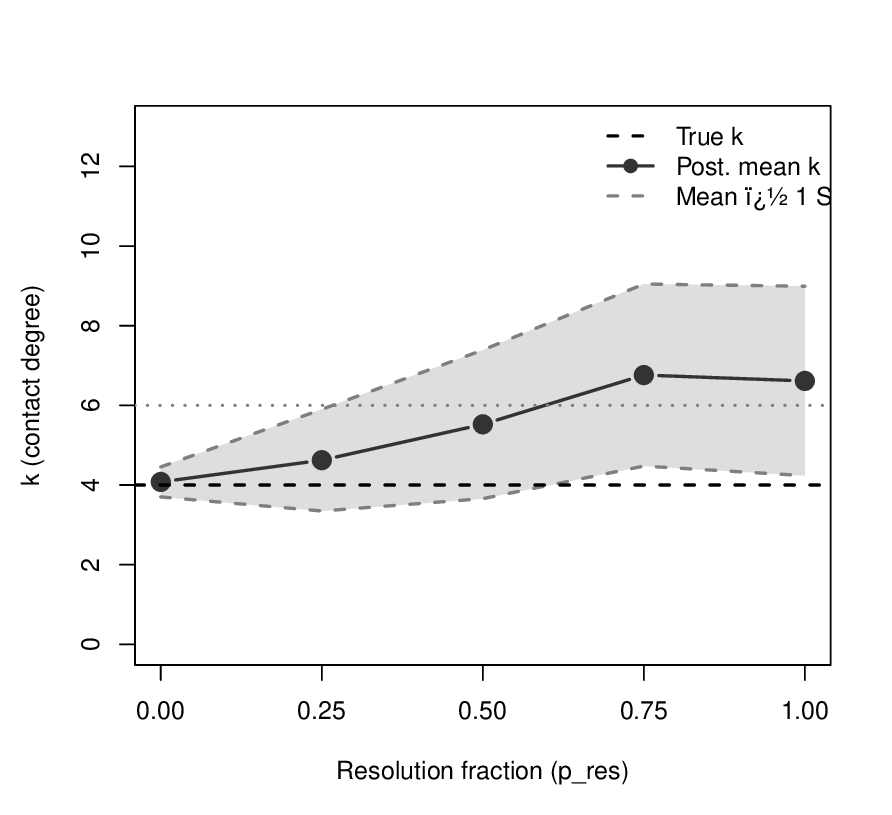}
\caption{Resolution sensitivity analysis with joint estimation of
$R_0$ and $k$.
Top left: posterior densities of $R_0$ at $p_{\mathrm{res}} \in
\{0.0,\,0.5,\,1.0\}$.
Top right: posterior mean $\hat{R}_0$ and 95\% HPDI versus
$p_{\mathrm{res}}$; the dashed horizontal line marks the true
$R_0 = 6.0$.
Bottom left: posterior PMF of $k$ at the same three levels.
Bottom right: posterior mean of $k$ and one-SD band versus
$p_{\mathrm{res}}$; the dashed horizontal line marks the true $k = 4$.
Fixed parameters: $\gamma = 1.0$, $p_{\mathrm{obs}} = 0.5$.}
\label{fig:partial_sweep}
\end{figure}

\begin{table}[ht]
\centering
\footnotesize
\caption{Posterior summaries of $R_0$ and $k$ across resolution
fractions $p_{\mathrm{res}}$, defined as the proportion of internal
branching times treated as unknown and latent in the MCMC.\@
At $p_{\mathrm{res}} = 0$ all branching times are known and fixed at their true
values; at
$p_{\mathrm{res}} = 1$ only tips are observed and all branching times are unknown.
Fixed: $\gamma = 1.0$, $p_{\mathrm{obs}} = 0.5$.
True $R_0 = 6.0$, true $k = 4$.
Full tree MLE ($k = 4$ AIC-selected, $N = 753$ trees) included as a reference.
}
\label{tab:partial_results}
\begin{tabular}{crrrrrrccc}
\toprule
$p_{\mathrm{res}}$ & \multicolumn{6}{c}{Posterior for $R_0$} &
  \multicolumn{3}{c}{Posterior for $k$} \\
\cmidrule(lr){2-7}\cmidrule(lr){8-10}
 & Mean & SD & \multicolumn{2}{c}{95\% HPDI} & Width & Bias
 & Mean & SD & Mode \\
\cmidrule(lr){4-5}
 & & & Lower & Upper & & & & & \\
\midrule
0.00 & 6.52 & 0.80 & 4.99 & 8.07 & 3.08 & $+0.52$ & 4.08 & 0.38 & 4 \\
0.25 & 5.86 & 0.79 & 4.33 & 7.55 & 3.22 & $-0.14$ & 4.62 & 1.28 & 4 \\
0.50 & 5.49 & 0.78 & 4.19 & 7.00 & 2.81 & $-0.52$ & 5.52 & 1.87 & 4 \\
0.75 & 5.04 & 0.69 & 3.85 & 6.40 & 2.55 & $-0.96$ & 6.76 & 2.29 & 5 \\
1.00 & 5.19 & 0.71 & 4.04 & 6.69 & 2.65 & $-0.81$ & 6.61 & 2.38 & 5 \\
\midrule
Full tree MLE & 6.03 & --- & 5.73 & 6.33 & 0.60 & $+0.03$ &
  \multicolumn{3}{c}{$k=4$ (AIC-selected)} \\
\bottomrule
\end{tabular}
\end{table}

\subsubsection{Results}
\label{subsubsec:resolution_results}

Table~\ref{tab:partial_results} and Figure~\ref{fig:partial_sweep}
summarise the posteriors across the five resolution levels. Two
patterns emerge: $R_0$ is estimated with bias below one unit at every
level, while $k$ is recovered cleanly only when at least half of the
internal branching times are observed. The true value $R_0 = 6.0$
lies inside the 95\% highest posterior density interval (HPDI) at
all five $p_{\mathrm{res}}$, with the HPDI tightening from
$[4.99,\,8.07]$ at $p_{\mathrm{res}}=0$ to $[4.04,\,6.69]$ at
$p_{\mathrm{res}}=1$; the posterior mean drifts downward from $6.52$ to $5.19$, but the truth remains well inside the interval
throughout. The picture for $k$ is different: at
$p_{\mathrm{res}}\in\{0,\,0.25,\,0.50\}$ the posterior concentrates
at the true value $k=4$, with standard deviation rising only modestly
from $0.38$ to $1.87$, whereas at $p_{\mathrm{res}}\in\{0.75,\,1\}$
the mode shifts to $k=5$ and the posterior spreads (SD~$\approx 2.3$).

These patterns share a common origin. The observed branching tempo
constrains the product $k\beta$, so $R_0=k\beta/\gamma$ remains
identified throughout; separating $k$ from $\beta$ requires, in
addition, the spacing between successive transmissions, which
reveals how $(k-i)\beta$ changes as $i$ grows. When most branching
times are latent the likelihood surface becomes flat along the
$k\beta$ ridge, the posterior for $k$ inherits the shape of the
prior, and the secondary drift in $\hat R_0$ reflects larger $k$
being accepted with proportionally smaller $\beta$. For applications where only $R_0$ is of interest, the framework delivers useful
inference from very partial trees; for applications that also target contact degrees,
at least a partial resolution of the internal branching structure is required.

\section{Application to Estimation: Karnataka COVID-19 First Wave}
\label{sec:empirical-karnataka}

Having validated our estimator on simulated data, where it recovered the key parameters of the phylogenetic model with well-calibrated uncertainty, we now apply it to empirical contact tracing data from the first wave of COVID-19 in Karnataka, India.

\subsection{Dataset}
\label{subsec:dataset}

The data come from the Karnataka Integrated Disease Surveillance
Programme (IDSP) linelist for cases reported between 9 March and
31 May 2020, described in detail by \citet{gupta2022contact}.
The linelist records, for each confirmed case, a unique identifier,
the date of laboratory confirmation, the number of primary downstream
contacts the individual infected, a cluster
identifier linking individuals within the same transmission chain,
and the detection route (whether the individual presented
symptomatically or was identified through contact tracing).
Infector-to-infectee links are recorded explicitly in a separate
contact file of 5{,}105 directed edges, symptom onset dates are
available for 261 individuals from an auxiliary timing file, and
serial-interval data (infector onset minus infectee onset) are
available for 54 infector-infectee pairs.

\subsection{Cohort definition and timing}
\label{subsec:data-prep}

We restrict to the early phase, 9 March to 31 May 2020, yielding 3{,}221 confirmed cases. Our likelihood models individuals who enter surveillance through independent symptomatic confirmation; cases classified as Local Traced (identified solely because they were named
as a contact of a known case) are excluded because their
ascertainment is conditional on knowing their infector, which would inflate the effective sampling probability.
After exclusion of 820 traced cases (25.5\%), the analysis cohort
contains $n = 2{,}401$ individuals
(386 Local Untraced, 1{,}897 Imported Domestic, 118 Imported
International), giving an empirical sampling probability of
$p_{\mathrm{obs}} = 2401 / 3221 \approx 0.75$.
For each individual $i$, the observed type at detection is
$j_{\mathrm{obs},i}$, mapping directly
to the state in the multi-type branching process; $j_{\mathrm{obs}}$
ranges from 0 to 30 across the cohort, with 92.6\% of individuals
at $j_{\mathrm{obs}} = 0$ and mean $0.2782$ (See Table~\ref{tab:covid_jobs}).

The branch length entering the likelihood for individual $i$ is
$\delta_i = t_{\mathrm{conf},i} - \tau_{\mathrm{born},i}$.
The confirmation time $t_{\mathrm{conf},i}$ is observed for all
individuals; the infection time $\tau_{\mathrm{born},i}$ is
latent for 91.3\% of the cohort (onset dates are recorded
for only 209 of 2{,}401 individuals).
Latent infection times are proposed under the feasibility constraint $\tau_{\mathrm{born},i} \in [\ell_i, u_i]$
with $u_i = t_{\mathrm{conf},i} - \delta_{\min}$ (infection precedes
own confirmation by at least $\delta_{\min}$) and
$\ell_i = t_{\mathrm{conf},\mathrm{par}(i)}$ when the infector is
recorded or $\ell_i = 0$ for index cases.
The timescale is fixed by the mean serial interval
$\bar{s} = 5.07$ days (sd 4.91, $n = 54$), so that $\gamma = 1$
model unit equals $\bar{s}$ days; the mean onset-to-confirmation
delay of 5.15 days (sd 3.65, $n = 261$) is consistent with this
timescale.

\begin{table}[h]
    \centering
    \footnotesize
    \caption{Total number and frequencies of observed states at observation.}
    \label{tab:covid_jobs}
    \begin{tabular}{r*{22}{r}}
        \toprule
        $j_{obs}$ & 0 & 1 & 2 & 3 & 4 & 5 & 6 & 7 & 8 & 9 & 10 & 12 & 14 & 15 & 16 & 17 & 19 & 22 & 28 & 29 & 30 \\
        \midrule
        Freq. & 2223 & 75 & 31 & 16 & 15 & 14 & 6 & 2 & 3 & 1 & 3 & 1 & 1 & 1 & 1 & 1 & 3 & 1 & 1 & 1 & 1 \\
        \bottomrule
    \end{tabular}
\end{table}

\subsection{Parameter assumptions and priors}
\label{subsec:assumptions}

The fixed parameters are $\gamma = 1.0$ and $p_{\mathrm{obs}} = 0.75$.
Fixing $\gamma$ anchors the timescale, which is otherwise
non-identifiable due to invariance under simultaneous rescaling of
branch lengths and removal rate; fixing $p_{\mathrm{obs}}$ at the
empirical fraction is required because the profile likelihood in
$p_{\mathrm{obs}}$ is strictly monotone, as established in
Section~\ref{subsec_mle_known_tree_top}.

We model the contact degree as an independent draw
$K \sim \mathrm{NegBin}(\mu=\mu_K,\phi=\phi_K)$, $K\in\{1,\ldots,K_{\max}\}$,
the standard distributional assumption for respiratory pathogens
\citep{lloyd2005superspreading,endo2020estimating,adam2020clustering,gupta2022contact,okolie2023parameter}.
The dispersion is fixed at $\phi_K=0.29$, taken from
\citet[Table~S2]{gupta2022contact} for symptomatic cases ($n=203$)
in the same Karnataka cohort (95\% CI $0.23$--$0.37$); note that
the symbol $k$ in their notation is our $\phi_K$, whereas
$\mu_K=\mathbb{E}[K]$ is the free parameter that our MCMC estimates.
At each iteration $\mu_K$ determines the mixture weights
$w_k=\mathbb{P}(K=k\mid\mu_K,\phi_K)$ entering the ODE system and
the likelihood marginalisation over the unknown focal degree. We place priors $\mu_K\sim\mathrm{LogNormal}(\log 11,\,0.5)$
(median $11$, 95\% mass on $[4.1,\,29.3]$; the median matches both the reported median of 11 high-risk contacts per index case in
\citet{gupta2022contact} and the mean transmission cluster size
of $11.8$ in our cohort) and $R_0\sim\mathrm{LogNormal}(\log 3,\,1)$,
broad enough to cover the early-wave literature.

\subsection{Results}
\label{sec:results-covid}

For the first wave of the COVID-19 outbreak in Karnataka, India, our estimate indicates a posterior mean of $R_0 = 2.60$ with a 95\% HPDI of
$[1.08, 5.25]$ and a median of 2.33; the posterior mean of
$\mu_K$ is 17.5 with a posterior mode at $\mu_K = 10$.
Convergence diagnostics indicate adequate mixing
(Geweke $z = -1.91$, effective sample size $91$ from $800$ posterior
draws after burn-in and thinning).
Figure~\ref{fig:covid-posterior} shows the posterior distribution
of $R_0$ (top left), the posterior mass function for $\mu_K$
(top right), the joint posterior over $(\mu_K, R_0)$ (bottom
left), and the autocorrelation of the $R_0$ chain (bottom right).
The joint posterior shows that the marginal for $R_0$
remains informative because the data constrain the product directly. The posterior mean of $R_0 = 2.60$ for symptomatic cases sits within the range reported by other early wave studies of SARS-CoV-2 in
India.
\begin{table}[ht]
\centering
\footnotesize
\caption{Posterior summaries from the Karnataka COVID-19 MCMC.\@
Fixed parameters: $\gamma = 1.0$ (1 model unit $= 5.07$ days),
$p_{\mathrm{obs}} = 0.75$, $\phi_K = 0.29$.
MCMC: $N_{\mathrm{iter}} = 2{,}000$, burn-in $= 400$, thin $= 2$.\@}
\label{tab:covid-results}
\begin{tabular}{lrrrrr}
\toprule
Parameter & Mean & SD & Mode & \multicolumn{2}{c}{95\% HPDI} \\
\cmidrule(lr){5-6}
& & & & Lower & Upper \\
\midrule
$R_0$        & 2.60  & 1.26  & ---   & 1.08 & 5.25 \\
$\mu_K$      & 17.5  & 10.7  & 10.0   & ---  & ---  \\
\bottomrule
\end{tabular}
\end{table}

\begin{figure}[ht]
\centering
\includegraphics[width=0.34\linewidth]{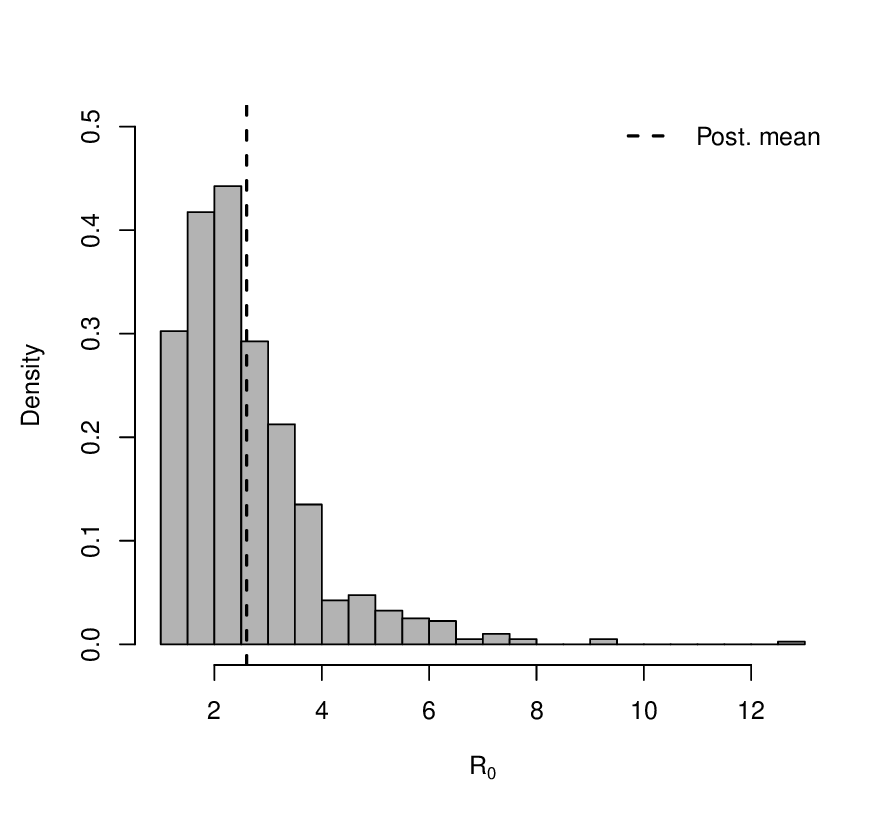}
\includegraphics[width=0.34\linewidth]{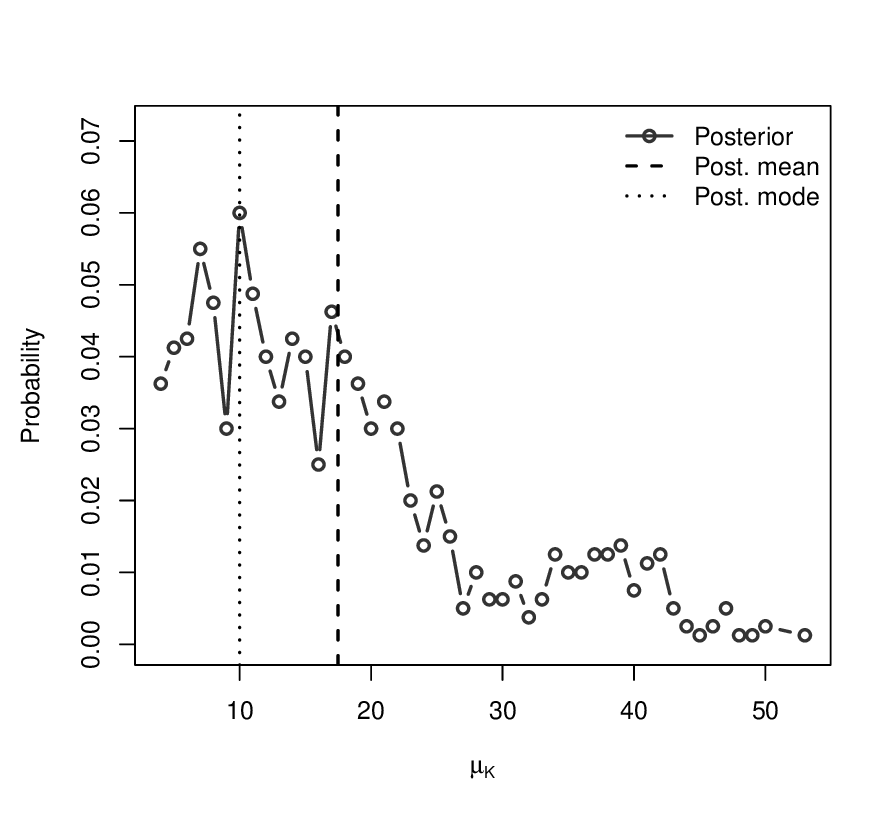}
\includegraphics[width=0.34\linewidth]{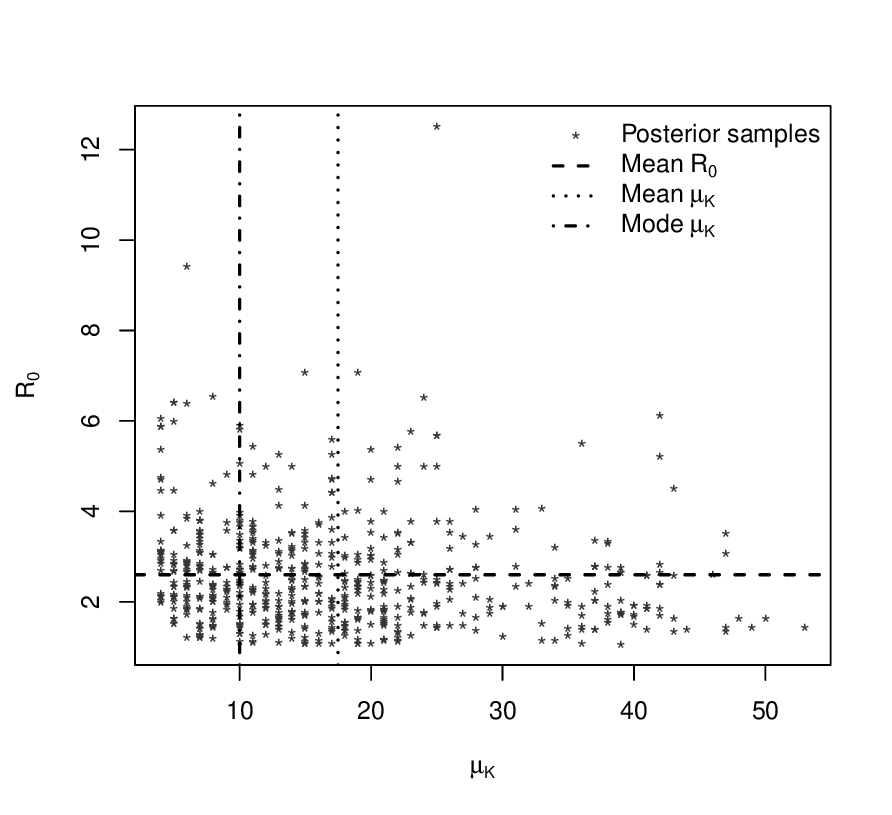}
\includegraphics[width=0.34\linewidth]{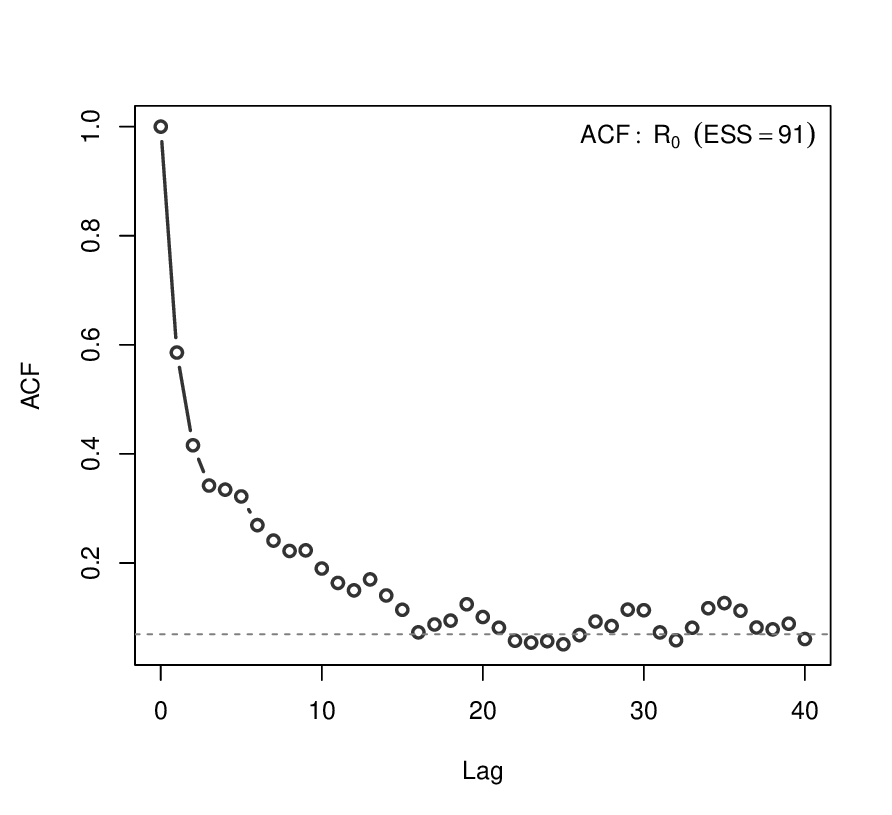}
\caption{Posterior inference for the Karnataka COVID-19 first wave.
Top left: marginal posterior density of $R_0$ (mean 2.60).
Top right: marginal posterior probability mass function of the
mean contact degree $\mu_K$ (mode at 10, mean 17.5).
Bottom left: joint posterior $(\mu_K, R_0)$.
Bottom right: autocorrelation of the $R_0$ chain, with effective
sample size $91$.
Fixed: $\gamma = 1.0$, $p_{\mathrm{obs}} = 0.75$, $\phi_K = 0.29$.}
\label{fig:covid-posterior}
\end{figure}

\citet{gupta2022contact} report $R = 2.04$ (95\% CI 1.56--2.67) for
symptomatic cases in the same Karnataka cohort, estimated from a
direct fit to the offspring distribution; our estimate is consistent with theirs while incorporating epidemiological information from the
contact graph and latent infection times. At national level, \citet{marimuthu2021covid} estimate
$R_0 \approx 1.83$ for early India, while exponential growth and
SEIR-based analyses of the pre-lockdown phase report values in the
range 2.0--2.6 \citep{ranjan2020covid, sardar2020assessment}.
Our estimate is consistent with this range and with the heterogeneous
contact structure characteristic of SARS-CoV-2 transmission, in which a small fraction of individuals contribute disproportionately to
onward spread.

\medskip
These point estimates fall into a three-level ordering of reproduction numbers implied by the model. The intrinsic basic reproduction number $R_0=2.60$ counts the secondary cases a freshly infected individual of typical degree would generate while all of its downstream contacts are still susceptible. The equilibrium-weighted residual reproduction number $\bar{R}_0$, defined in equation~\eqref{eq:Rbar0-random}, discounts the downstream contacts that are, on average, already infected once the type distribution has equilibrated. In the growing phase, approximately one downstream neighbour per lineage is already infected (Corollary~\ref{cor:Spi-random}), so $\bar{R}_0=R_0\,(1-1/\mu_K)\approx 2.45$ at the posterior mean contact degree $\mu_K\approx 17.5$. The empirical effective reproduction number of \citet{gupta2022contact}, $R=2.04$, measures the offspring actually realised in the field. Because the posterior for $\mu_K$ is multimodal and its mean is correspondingly unstable, we also evaluate $\bar{R}_0$ at the posterior mode $\mu_K=10$, obtaining $\bar{R}_0\approx 2.34$. The modal and mean contact degrees thus give very similar residual reproduction numbers, $2.34$ and $2.45$ respectively, and both fall between the empirical $R=2.04$ and the intrinsic $R_0=2.60$. The gap $R_0-\bar{R}_0$ is precisely the within-tree contact-depletion correction that our augmented $(i,k)$ state space is built to capture. In the growing phase an average lineage has one already-infected downstream neighbour (Corollary~\ref{cor:Spi-random}), so the gap reduces to $\beta/\gamma=R_0/\mu_K$ in the sense that it is governed by the mean contact degree alone. It is therefore small here, between roughly $0.15$ and $0.26$ across the mean and modal degree, simply because $\mu_K$ is large, so that losing a single susceptible contact barely lowers transmission potential. Furthermore, the offset of the empirical $R=2.04$ slightly below both $R_0$ and $\bar{R}_0$ should be read with caution, since $R_0$ itself is only weakly resolved (95\% HPDI $1.08$--$5.25$); to the extent that it is real, it reflects mechanisms outside the model's growing-phase assumptions, including non-pharmaceutical interventions, case isolation, and contact-tracing-driven quarantine in force during the first wave, together with our restriction to the traced symptomatic cohort.
We emphasise that these results are preliminary in the sense that the
analysis cohort is restricted to individuals with documented
contact tracing information, so the inference reflects the
underlying epidemiological dynamics in the traced sub-population rather than the
state-wide epidemic.

\section{Discussion}
\label{sec:discussion}

This work develops a likelihood framework that incorporates network heterogeneity into infectious disease inference
by treating contact degree as a structural feature that shapes both
transmission hazards and observation probabilities. Building on
multi-type branching process theory and applied directly to transmission trees, the framework links observed branching patterns to both epidemiological rates and network structure through a closed-form likelihood. It contributes to the broader effort of integrating network-based epidemic modelling with
likelihood-based tree inference, and complements rather than
reproduces the sequence-based phylodynamic literature in which
multi-type branching processes have so far been most prominent.
The framework uses an augmented state space
in which individuals are characterised by both their total contact
capacity and the number of contacts already infected. This allows the model to capture a fundamental epidemiological reality such that as an individual infects more of their contacts, fewer susceptible
pathways remain available, and transmission intensity naturally
declines. Rather than treating contact degree as a static parameter,
we show how it directly shapes the likelihood of observed branching
patterns through both survival probabilities and tip densities
derived from closed-form differential equations. We derive exact formulas for fixed deterministic contact degrees and then generalise to random degrees. Finally, we validate the estimators on simulated transmission trees, and demonstrate its practical utility through an application to first wave COVID-19 contact tracing data from Karnataka, India.

\medskip
Our findings reveal a clear
relationship between tree resolution and parameter identifiability
in the sense that the basic reproduction number remains strongly
estimable at all resolution levels. This is because the product of
degree and transmission rate is fundamentally pinned by the
observed branching tempo regardless of how uncertainty is
distributed between these two components. The posterior for $R_0$
concentrates around its true value and maintains well calibrated
credible intervals whether internal branching times are fully
observed, partially latent or entirely unobserved. Even in the
most extreme scenario, where we have information about the tips
only, the posterior exhibits negligible bias and contains the truth
comfortably within its uncertainty bounds. These results follow
directly from the likelihood structure, in which inter-branch
survival and infection rates dominate the information content.
Secondary terms provide additional information sufficient to
constrain $R_0$, even when contact degree itself remains diffuse. From a
practitioner's perspective, reliable inference of epidemic intensity
is achievable even when network structure cannot be fully resolved.

\medskip
A central assumption underlying our likelihood is that the contact
structure is a static tree in which every observed transmission
can be uniquely assigned to an ancestor in the reconstructed
transmission tree, and that an individual's contact degree is fixed
throughout their infectious period. This assumption keeps the
likelihood closed-form but can be violated in real epidemics in
two related ways. First, outside infections occur when a lineage
is infected not by its apparent ancestor in the recorded tree but by
another source elsewhere in the network; simulation
studies~\citep{okolie2023parameter, metzig2019phylogenies} show
that such events inflate the apparent number of secondary cases
at internal nodes and can bias estimates of transmission rates
and degree variance upward. Second, dynamic contact patterns such as edge creation and deletion during the infectious period mean that the effective degree of a focal individual changes over
time, which is particularly relevant for infections with long
infectious periods or in endemic settings where social contacts
evolve on a timescale comparable to disease
spread~\citep{okolie2020exact}.

Both violations point to the same underlying generalisation in the sense that the
contact graph is more naturally modelled as a network in which reticulate patterns such as multiple infection
pathways and time-varying connectivity can be accommodated.
Evidence from related
work~\citep{okolie2023parameter} indicates that our estimator
remains stable under moderate deviations from the strict tree
assumption, but its accuracy diminishes when such events become
frequent. Extending the framework to a dynamic contact graph formulation, either through explicit correction factors for outside
infections or by treating the degree as a stochastic process
evolving alongside the epidemic, is an important direction for
future work.

\medskip
This work opens a pathway for
richer inference in disease outbreaks where detailed transmission
information is available independently of, or alongside, molecular
sequencing. The Karnataka application
illustrates this concretely. Applied to confirmed symptomatic cases
from the first wave (March--May 2020) under a Negative Binomial
contact degree distribution with dispersion fixed from the
empirical estimate of \citet{gupta2022contact}, the framework
produces a reproduction number consistent with independent early-wave estimates
for the same outbreak~\citep{gupta2022contact} and for India more
broadly~\citep{marimuthu2021covid, ranjan2020covid}. Rather than
treating network structure and tree topology as separate
sources of evidence, the framework integrates them into a unified
likelihood. This is particularly valuable for pathogens like
SARS-CoV-2 in outbreak settings where aggressive contact tracing
generates detailed transmission information that is often available
before, or independently of, sequence-based reconstruction.
The applicability extends naturally to scenarios with
partially resolved trees, incomplete contact tracing data, and
heterogeneous observation effort. Furthermore, the assumption of
constant removal rate and exponential growth could be relaxed in future extensions,
time-varying transmission rates could be incorporated to account for interventions,
and untraced contacts or missing edges could be brought explicitly into the likelihood.

\medskip
All in all, this study demonstrates that integrating network
structure into likelihood-based tree inference yields richer, more
interpretable estimates of epidemic parameters than homogeneous
mixing assumptions allow. More fundamentally, it reveals a clear
relationship between tree resolution and network
identifiability, showing that observed internal structure enhances
our ability to estimate contact patterns. As contact
tracing data continue to grow in richness and scale, integrating
network structure with likelihood-based inference on transmission trees will become
increasingly central to outbreak investigation and response.

\par\bigskip
{\bf Data Availability Statement}
\par\medskip
The empirical COVID-19 contact-tracing data from Karnataka analysed in this study were derived from the publicly available Karnataka Integrated Disease Surveillance Program (IDSP) linelist and are available via the GitHub repository \url{https://github.com/CovidToday/covid19-karnataka}, as described by Gupta et al. (2022). The simulated transmission trees and all underlying computational code, including the implementation of the ordinary differential equation solver and the Markov Chain Monte Carlo (MCMC) estimation framework, are openly available in the GitHub repository at \url{https://github.com/Austine316/Multi-type-branching-inference-on-contact-trees}.

\par\bigskip
{\bf Funding}
\par\medskip
This research was supported by a grant from the German Academic Exchange Service (DAAD) (Augustine Okolie) and by the Deutsche Forschungsgemeinschaft (DFG) through the TUM International Graduate School of Science and Engineering (IGSSE), GSC 81, within the project GENOMIE\_QADOP (Johannes M\"uller). Eno Akarawak and Isaac Ajiboye acknowledge support from the DAAD SDG Partnership Project, the Competent Network Mathematical Epidemiology for sub-Saharan Africa (CoNeMesSA).

\par\bigskip
{\bf Acknowledgements}
\par\medskip
The author thanks Simran Talreja
for her contributions to this project during her Master’s thesis at the Technical University of Munich.

\par\bigskip
{\bf Declaration of Competing Interest}
\par\medskip
{None}

\newpage
\appendix

\section*{Supplementary Information}

Detailed proofs of the equilibrium frequency distributions, convergence rates, expected downstream infection counts, and simulation strategy are provided in the accompanying Supplementary Information (SI) file.

\bibliography{references}

\bibliographystyle{unsrtnat}

\end{document}